\documentclass[]{acmart}
\usepackage{xspace}
\usepackage{ifthen}
\usepackage{amsthm}
\usepackage{subfig}
\usepackage{soul}
\usepackage{multirow}
\acmConference{}{}{} 
\renewcommand\footnotetextcopyrightpermission[1]{} 

\title{Fully Read/Write Fence-Free Work-Stealing with Multiplicity}

\author{Armando Casta\~neda}
\affiliation{
  \institution{Instituto de Matem\'aticas, UNAM}
  \country{M\'exico}
}
\email{armando.castaneda@im.unam.mx}

\author{Miguel Pi\~na}
\affiliation{
  \institution{Facultad de Ciencias, UNAM}
  \country{M\'exico}
}
\email{miguel\_pinia@ciencias.unam.mx}

\newcommand{\Put}{{\sf Put}\xspace}
\newcommand{\Take}{{\sf Take}\xspace}
\newcommand{\Puts}{{\sf Puts}\xspace}
\newcommand{\Takes}{{\sf Takes}\xspace}
\newcommand{\Steal} {{\sf Steal}\xspace}

\newcommand{\RAW}{{\sf Read-After-Write}\xspace}
\newcommand{\RMW}{{\sf Read-Modify-Write}\xspace}
\newcommand{\CAS}{{\sf Compare\&Swap}\xspace}
\newcommand{\TAS}{{\sf Test\&Set}\xspace}
\newcommand{\R}{{\sf Read}\xspace}
\newcommand{\W}{{\sf Write}\xspace}
\newcommand{\MaxReg}{{\sf MaxRegister}\xspace}
\newcommand{\RangeMaxReg}{{\sf RangeMaxRegister}\xspace}

\newcommand{\SWAP}{{\sf Swap}\xspace}
\newcommand{\Enq}{{\sf Enqueue}\xspace}
\newcommand{\Push}{{\sf Push}\xspace}
\newcommand{\Pop}{{\sf Pop}\xspace}
\newcommand{\Deq}{{\sf Dequeue}\xspace}
\newcommand{\op}{{\sf op}\xspace}
\newcommand{\true}{{\sf true}\xspace}
\newcommand{\false}{{\sf false}\xspace}
\newcommand{\epty}{{\sf empty}\xspace}

\newcommand{\SetLin}{{\sf SetLin}\xspace}
\newcommand{\WFWSM}{{\sf WS-MULT}\xspace}
\newcommand{\MaxR}{{\sf MaxRead}\xspace}
\newcommand{\MaxW}{{\sf MaxWrite}\xspace}
\newcommand{\RMaxR}{{\sf RMaxRead}\xspace}
\newcommand{\RMaxW}{{\sf RMaxWrite}\xspace}
\newcommand{\NCWSM}{{\sf WS-WMULT}\xspace}

\newcommand{\BNBWSM}{{\sf B-WS-MULT}\xspace}
\newcommand{\BNCWSM}{{\sf B-WS-WMULT}\xspace}

\newcounter{linecounter}
\newcommand{\linenumbering}{\ifthenelse{\value{linecounter}<10}{(0\arabic{linecounter})}{(\arabic{linecounter})}}
\renewcommand{\line}[1]{\refstepcounter{linecounter}\label{#1}\linenumbering}
\newcommand{\resetline}[1]{\setcounter{linecounter}{0}#1}

\newtheorem{remark}{Remark}[section]

\definecolor{hanblue}{rgb}{0.27, 0.42, 0.81}

\begin{document}

\begin{abstract}

\emph{Work-stealing} is a popular technique to implement dynamic
\emph{load balancing} in a distributed manner.
In this approach, each process \emph{owns} a set of tasks that have to
be executed. The \emph{owner} of the set can put tasks in it and can
take tasks from it to execute them. When a process runs out of tasks,
instead of being idle, it becomes a \emph{thief} to steal tasks from a
\emph{victim}. Thus, a work-stealing algorithm provides three
high-level operations: \Put and \Take, which can be invoked only by
the owner, and \Steal, which can be invoked by a thief.

One of the main targets when designing work-stealing algorithms is to
make \Put and \Take as simple and efficient as possible.
Unfortunately, it has been shown that any work-stealing algorithm in
the standard asynchronous model must use expensive \RAW
synchronization patterns or atomic \RMW instructions,
which may be costly in practice.  Thus, prior research has
proposed \emph{idempotent} work-stealing, a relaxation for which there
are algorithms with \Put and \Take devoid of \RMW atomic instructions
and \RAW synchronization patterns; however, \Put uses fences among \W
instructions, and \Steal uses \CAS and fences among \R instructions.
In the TSO model, in which \W (resp. \R) instructions cannot be
reordered, there have been proposed fully fence-free work-stealing
algorithms whose \Put and \Take have similar properties but \Steal
uses \CAS or a \emph{lock}.

This paper considers work-stealing with \emph{multiplicity}, a
relaxation in which every task is taken by \emph{at least} one
operation, with the requirement that any process can extract a task
\emph{at most once}.  Two versions of the relaxation are
considered and two \emph{fully \R/\W} algorithms are presented in the
standard asynchronous shared memory model, both devoid of \RAW
synchronization patterns in all its operations, the second algorithm
additionally being \emph{fully fence-free}.  Furthermore, the
algorithms have logarithmic and constant step complexity,
respectively.  To our knowledge, these are the first algorithm for
work-stealing possessing all these properties.  Our algorithms are
also wait-free solutions of relaxed versions of single-enqueue
multi-dequeuer queues, namely, with multiplicity and weak
multiplicity.  The algorithms are obtained by reducing work-stealing
with multiplicity and weak multiplicity to \MaxReg and \RangeMaxReg, a
relaxation of \MaxReg which might be of independent interest.

An experimental evaluation shows that our fully fence-free
algorithm exhibits better performance than The Cilk, Chase-Lev and
Idempotent Work-Stealing algorithms, while a simple variant of it,
with a \SWAP-based \Steal operation, shows a lower performance
but keeps competitive.
\end{abstract}

\maketitle
\section{Introduction}

\subsubsection*{\bf Context}

\emph{Work-stealing} is a popular technique to implement dynamic
\emph{load balancing} in a distributed manner; it has been used in
several contexts, from programming languages and parallel-programming
frameworks to SAT solver and state space search exploration in model
checking (e.g.~\cite{ACDHLMTUZ09, BJKLRZ95, CGSDKEPS05, FDSZ01, FLR98,
L00, RRPBK07}).

In work-stealing, each process \emph{owns} a set of tasks that have to
be executed.  The \emph{owner} of the set can put tasks in it and can
take tasks from it to execute them. When a process runs out of tasks
(i.e. the set is empty), instead of being idle, it becomes a
\emph{thief} to steal tasks from a \emph{victim}. Thus, a
work-stealing algorithm provides three high-level operations: \Put and
\Take, which can be invoked only by the owner, and \Steal, which can
be invoked by a thief.

One of the main targets when designing work-stealing algorithms is to
make \Put and \Take as simple and efficient as possible.
Unfortunately, it has been shown that any work-stealing algorithm in
the standard asynchronous model must use \RAW synchronization patterns
or atomic \RMW instructions (e.g. \CAS or \TAS)~\cite{AGHKMV11}.  \RAW
synchronization patterns are based on the \emph{flag
principle}~\cite{HS08} (i.e. writing on a shared variable and then
reading another variable), and hence they require the use of memory
\emph{fences} (also called \emph{barriers}) so that the read and write
instructions are not reordered by a compiler; it is well-known that
fences that avoid read and writes to be reordered are highly costly in
real multicore architectures.  While atomic \RMW instructions, that
have high coordination power (which can be formally measured through the
\emph{consensus number} formalism~\cite{H91}), are in principle slower than the
simple \R/\W instructions.\footnote{In practice, contention
might be the dominant factor, namely, an uncontended \RMW instruction
can be faster than contended \R/\W instructions.}  Indeed, the known
work-stealing algorithms in the literature are based on the flag
principle in their \Take/\Steal operations~\cite{CL05, FLR98, HLMS06,
HS02}.  Thus, a way to circumvent the result in~\cite{AGHKMV11} is to
consider work-stealing with relaxed semantics, or to make extra
assumptions on the model.  As far as we know,~\cite{MlVS09}
and~\cite{MA14} are the only works that have followed these
directions.

\emph{Idempotent} work-stealing~\cite{MlVS09} relaxes the semantics
allowing a task to be taken \emph{at least once}, instead of
\emph{exactly once}. This relaxation is useful in contexts where it is
ensured that no task is repeated (e.g. by checking first whether a
task is completed) or the nature of the problem solved tolerates
repeatable work (e.g. parallel SAT solvers).  Three idempotent
work-stealing algorithms are presented in~\cite{MlVS09}, that
insert/extract tasks in different orders (FIFO, LIFO and double-ended queue).
The relaxation allows each
of the algorithms to circumvent the impossibility result
in~\cite{AGHKMV11}, only in its \Put and \Take operations as they use
only \R/\W instructions and are devoid of \RAW synchronization
patterns; however, \Steal uses \CAS. Moreover, the algorithms are not
fully fence-free: \Put uses fences among \W instructions and \Steal
uses fences among \R instructions.  As for progress guarantees, \Put and \Take
are \emph{wait-free} (i.e. each invocation always terminates) while
\Steal is only \emph{nonblocking} (i.e. an invocation may be blocked
by the progress of other invocations).

It is presented in~\cite{MA14} two \emph{fully fence-free}
work-stealing algorithms in the TSO model~\cite{SSONM10}, whose \Put
and \Take operations use only \R/\W instructions but \Steal either uses \CAS or a
\emph{lock}; \Put and \Take are wait-free, and \Steal is either
nonblocking (in the \CAS-based algorithm) or \emph{blocking} (in the
lock-based algorithm).  The algorithms are clever adaptations of the
well-known THE Cilk and Chase-Lev work-stealing algorithms~\cite{CL05,
FLR98} to the TSO model.  Generally speaking, in this model \W
(resp. \R) instructions cannot be reordered, hence fences among \W
(resp. \R) instructions are not needed; additionally, each process has
a local buffer where its \W instructions are stored until they are
eventually propagated to the main memory (in FIFO order).  To avoid
\RAW patterns, it is considered in~\cite{MA14} that buffers are of
bounded size.

In this paper we are interested in the following theoretical question:
whether there are meaningful relaxations of
work-stealing that allows us to design fully \R/\W, fence-free
and wait-free algorithms, i.e. using only \R/\W instructions
(which lay at the lowest level of the consensus number hierarchy~\cite{H91})
without fences, providing strong termination guarantees and
with no extra assumptions to the model of computation.
In other words, if simple synchronization
mechanisms suffice to solve non-trivial relaxations of work-stealing.

\subsubsection*{\bf Contributions}

First, we consider work-stealing with
\emph{multiplicity}~\cite{CRR20}, a relaxation in which every task is
taken by \emph{at least} one operation, and, differently from
idempotent work-stealing, it requires that if a task is taken by
several \Take/\Steal operations, then they must be \emph{pairwise
concurrent}; therefore, no more than the number of processes in the
system can take the same task.  In our relaxation, tasks are
inserted/extracted in FIFO order.  We present a fully \R/\W algorithm
for work-stealing with multiplicity, with its \Put operation being
fence-free and its \Take and \Steal operations being devoid of \RAW
synchronization patterns.  As for progress, all operations are
wait-free with \Put having constant \emph{step complexity} (i.e. the
maximum number of steps/instructions needed to complete) and \Take and
\Steal having logarithmic step complexity.  The algorithm stands for
its simplicity, based on a single instance of a \MaxReg object, hence
showing that work-stealing with multiplicity reduces to \MaxReg.  A
\MaxReg object~\cite{AAC12,JTT00} provides two operations: \MaxR that
returns the maximum value written so far in the object, and \MaxW that
writes a new value only if it is greater than the largest value that
has been written so far.

We also consider a variation of work-stealing with multiplicity in
which \Take/\Steal operations extracting the same task \emph{may not
be concurrent}, however, each process extracts a task \emph{at most
once}; this variant, which insert/extract tasks in FIFO order too, is
called work-stealing with \emph{weak multiplicity}.  For this
relaxation, we present an algorithm, inspired in our first solution,
which uses only \R/\W instructions, is \emph{fully fence-free} and
all its operations are wait-free; furthermore each operation executes
a constant number of instructions to complete.  To our knowledge, this
is the first algorithm for work-stealing having all these properties.
The algorithm is obtained by a reduction to a relaxation of \MaxReg
with a \RMaxR operation that returns a value in a \emph{range} of
values that have been written in the register; the range always
includes the maximum value written so far.  This relaxation, which
might be of independent interests, is called \RangeMaxReg.

Additionally, we show that each of our algorithms can be easily
modified so that every task is extracted by a bounded number of
operations; more specifically, by at most one \Take operation and at
most one \Steal operation.  In the modified algorithms, \Put and \Take
remain the same, and a single \SWAP instruction is added to \Steal,
which remains fence-free.

We stress that our algorithms are also wait-free solutions of
relaxed versions of single-enqueue multi-dequeuer queues, namely,
with multiplicity and weak multiplicity.
To the best of our knowledge, together with idempotent FIFO,
these are the only single-enqueue multi-dequeuer queue relaxations
that have been studied.

Formal specifications and correctness proofs are provided, using the
\emph{linearizability}~\cite{HW90} and \emph{set-linearizable}
correctness formalisms~\cite{CRR18, N94}.  Intuitively,
set-linearizability is a generalization of linearizability in which
several concurrent operations are allowed to be linearized at the same
linearization point.

Work-stealing with multiplicity and idempotent work-stealing are
closely related, but they are not the same.  We observe that the
idempotent work-stealing algorithms in~\cite{MlVS09} allow a task to
be extracted by an unbounded number of \Steal operations; moreover, a
thief can extract the same task an unbounded number of times.  This
observation implies that the algorithms in~\cite{MlVS09} do not solve
work-stealing with multiplicity.  Therefore, the relaxations and
algorithms proposed here provide stronger guarantees than idempotent
work-stealing algorithms, without the need of heavy synchronization
mechanisms.

We complement our results by studying the question if
there are implementations of our algorithms with good performance in
practical settings. We conducted an experimental evaluation comparing
our algorithms to THE Cilk, Chase-Lev and the idempotent work-stealing
algorithms.  In the experiments, the fully fence-free
algorithm exhibits a better performance than the previously mentioned
algorithms, while its bounded version, with a \SWAP-based \Steal
operation, shows a lower performance but keeps competitive.

\subsubsection*{\bf Related Work}

To the best of our knowledge,~\cite{MlVS09} and~\cite{MA14} are the
only works that have been able to avoid costly synchronization
mechanisms in work-stealing algorithms.  The three idempotent
algorithms in~\cite{MlVS09} insert/extract tasks in FIFO and LIFO
orders, and as in a double-ended queue (the owner working on one side
and the thieves on the other).  The experimental evaluation shows that
the idempotent algorithms outperform THE Cilk and Chase-Lev
work-stealing algorithms. The work-stealing algorithms in~\cite{MA14}
are fence-free adaptations of THE Cilk and Chase-Lev~\cite{CL05,
FLR98} to the TSO model. The model itself guarantees that fences among
\W (resp. \R) instructions are not needed, and the authors
of~\cite{MA14} assume that write-buffers in TSO are bounded so that a
fence-free coordination mechanism can be added to the \Take and \Steal
operations of THE Cilk and Chase-Lev.  The experimental evaluation
in~\cite{MA14} shows that their algorithms outperform THE Cilk and
Chase-Lev and sometimes achieve performance comparable to the
idempotent work-stealing algorithms.

The notion of multiplicity was recently introduced in~\cite{CRR20} for
queues and stacks.  In a queue/stack with multiplicity, an item can be
dequeued/popped by several operations but only if they are concurrent.
\R/\W set-linearizable algorithms for queues and stacks with
multiplicity and without \RAW synchronization patterns are provided
in~\cite{CRR20}, and it is noted that these algorithms are also
solutions for work-stealing with multiplicity; the step complexity of
\Enq/\Push, which implements \Put, is $\Theta(n)$ while the step
complexity of \Deq/\Pop, which implements \Take and \Steal, is unbounded,
where $n$ denotes the number of processes.
Our algorithms for work-stealing with multiplicity follows
different principles than those in~\cite{CRR20}.
The weak multiplicity relaxation of work-stealing
follows the same idea as the relaxations in~\cite{CRR20} as it requires that
any algorithm provides ``exact'' responses in sequential executions,
namely, the relaxation applies only if there is no contention.

As far we know, only two single-enqueue multi-dequeuer queue
algorithms have been proposed so far~\cite{D04, JP05}.  The algorithm
in~\cite{D04} is wait-free with constant step complexity and uses \RMW
instructions with consensus number~2, specifically, \SWAP and {\sf
Fetch\&Increment}; however, the algorithm uses arrays of infinite
length. It is explained in~\cite{D04} that this assumption can be
removed at the cost of increasing the step complexity to $O(n)$, where
$n$ denotes the number of processes.  The algorithm in~\cite{JP05} is
wait-free with $O(\log n)$ step complexity and uses {\sf LL/SC}, a
\RMW instruction whose consensus number is $\infty$, and using a
bounded amount of memory.  Our results show that there are
single-enqueue multi-dequeuer queue relaxations that can be solved
using very light synchronization mechanisms.

\subsubsection*{\bf Organization}

The rest of the paper is organized as follows.
Section~\ref{sec-preliminaries} describes the model of computation and
the linearizability and set-linearizability formalisms.
Section~\ref{sec-ws-mult} formally defines work-stealing with
multiplicity and presents a solution based on a single \MaxReg object.
The work-stealing with weak multiplicity and its algorithm are
presented in Section~\ref{sec-ws-nc-mult}.  Simple variants of the
algorithms that bound the number of operations that can extract the
same task are discussed in Section~\ref{sec-bound-mult}, while
Section~\ref{sec-removing-infinite-arrays} shows how our algorithms
can be implemented using finite length arrays.
Section~\ref{sec-idem-neq-mult} explain the differences between
work-stealing with multiplicity and idempotent work-stealing.
Sections~\ref{sec-experiments} presents the results of the experimental
evaluation.  The paper concludes with a final discussion
in~Section~\ref{sec-final-discussion}.

\section{Preliminaries}
\label{sec-preliminaries}

\subsubsection*{\bf Model of Computation}

We consider a standard concurrent shared memory system with $n \geq 2 $
\emph{asynchronous} processes, $p_0, \hdots, p_{n-1}$, which may
\emph{crash} at any time during an execution.
Processes communicate
with each other by invoking \emph{atomic} instructions of \emph{base}
objects: either simple \R/\W instructions, or more powerful \RMW
instructions, such as \SWAP or \CAS. The \emph{index} of process $p_i$
is $i$.

A \emph{concurrent object}, or \emph{data type}, is, roughly speaking,
defined by a state machine consisting of a set of states, a finite set
of operations, and a set of transitions between states. The
specification does not necessarily have to be sequential, namely,
state transitions might involve several invocations.  The following
subsection formalize \emph{sequential} and \emph{set-sequential}
objects.

An \emph{algorithm} for a concurrent object $T$ is a distributed
algorithm $\mathcal A$ consisting of local state machines $A_1,
\hdots, A_n$.  Local machine $A_i$ specifies which instructions of
base objects $p_i$ executes in order to return a response, when it
invokes a (high-level) operation of $T$; each of these instructions is
a \emph{step}.

An \emph{execution} of $\mathcal A$ is a (possibly infinite) sequence
of steps, namely, instructions of base objects, plus invocations and
responses of (high-level) operations of the concurrent object $T$,
with the following properties:

\begin{enumerate}
  \item Each process first invokes an operation, and only when it has
    a corresponding response, it can invoke another operation, i.e.,
    executions are \emph{well-formed}.
  \item For any invocation to an operation {\sf op}, denoted $inv({\sf
      op})$, of a process $p_i$, the steps of $p_i$ between that invocation
    and its corresponding response (if there is one), denoted $res({\sf
      op})$, are steps that are specified by $\mathcal A$ when $p_i$ invokes
    \op.
\end{enumerate}

An operation in an execution is \emph{complete} if both its invocation
and response appear in the execution.  An operation is \emph{pending}
if only its invocation appears in the execution.  A process is
\emph{correct} in an execution if it takes infinitely many steps.  An
execution $E$ is an \emph{extension} of an execution $F$, if $E$ is a
prefix of $F$, namely, $E = F F'$ for some $F'$.

An algorithm, or an operation of it, is \emph{nonblocking} if whenever
processes take steps, at least one of the invocations
terminates~\cite{HW90}.  Formally, in every infinite execution,
infinitely many invocations are completed.  An algorithm, or an
operation of it, is \emph{wait-free} if every process completes each
operation in a finite number of its steps~\cite{H91}.  Formally, if a
process executes infinitely many steps in an execution, all its
invocations are completed.  \emph{Bounded wait-freedom}~\cite{H91a}
additionally requires that there is a bound on the number of steps
needed to terminate.  Thus, a wait-free implementation is nonblocking
but not necessarily vice versa.

The \emph{step complexity} of an operation of an algorithm is the
maximum number steps a process needs to execute in order to return.

In a \RAW synchronization pattern, a process first writes in a shared
variable and then reads another shared variable, maybe executing other
instructions in between.  Note that \RAW patterns can be avoided if in
each (high-level) operation of an algorithm, a process performs a
sequence of writes, or a sequence of reads followed by a sequence of
writes.

We say that an algorithm, or one of its operations, is
\emph{fence-free} if it does not require any specific ordering among
its steps, beyond what is implied by data dependence. Thus, a
fence-free algorithms does not use \RAW synchronization patterns.  In
our algorithms, we use the notation $\{O_1.inst_1, \hdots,
O_x.inst_x\}$ to denote that the instructions $O_1.inst_1, \hdots,
O_x.inst_x$ can be executed in any order, hence no fence is required
between any pair of them.

For sake of simplicity in the analysis, first we will present our
algorithms using arrays of infinite length, and later explain how to
implement them using arrays of finite length.

\subsubsection*{\bf Correctness Conditions}

\emph{Linearizability}~\cite{HW90} is the standard notion used to
identify a correct implementation.  Intuitively, an execution is
linearizable if its (high-level) operations can be ordered
sequentially, without reordering non-overlapping operations, so that
their responses satisfy the specification of the implemented object.

A \emph{sequential specification} of a concurrent object $T$ is a
state machine specified through a transition function $\delta$. Given
a state $q $ and an invocation $inv({\sf op})$, $\delta(q, inv({\sf
op}))$ returns the tuple $(q', res({\sf op}))$ (or a set of tuples if
the machine is \emph{non-deterministic}) indicating that the machine
moves to state $q'$ and the response to \op is $res({\sf op}$).  The
sequences of invocation-response tuples, $\langle inv({\sf op}):
res({\sf op}) \rangle$, produced by the state machine are its
\emph{sequential executions}.

To formalize linearizability we define a partial order $<_E$ on the
completed operations of an execution $E$: ${\sf op} <_E {\sf op}'$ if
and only if $res({\sf op})$ precedes $inv({\sf op}')$ in $E$.  Two
operations are \emph{concurrent}, denoted ${\sf op} \, ||_E \, {\sf op}'$,
if they are incomparable by $<_E$.  The execution is \emph{sequential}
if $<_E$ is a total order.

\begin{definition}[Linearizability] Let $\mathcal A$ be an algorithm
for a concurrent object $T$.  A  finite execution $E$ of $\mathcal A$ is
\emph{linearizable} if there is a sequential execution $S$ of $T$ such
that,

\begin{enumerate}
  \item $S$ contains every completed operation of $E$ and might
    contain some pending operations.  Inputs and outputs of invocations
    and responses in $S$ agree with inputs and outputs in $E$,
  \item for every two completed operations {\sf op} and ${\sf op}'$ in
    $E$, if ${\sf op} <_E {\sf op}'$, then {\sf op} appears before ${\sf
      op}'$ in $S$.
\end{enumerate}

We say that $\mathcal A$ is \emph{linearizable} if each of its
finite executions is linearizable.
\end{definition}

Roughly speaking, while linearizability requires a total order on the
operations, set-linearizability~\cite{CRR18,N94} allows several
operations to be linearized at the same linearization point.
Figure~\ref{fig-example-linear} schematizes the differences between
the two consistency conditions where each double-end arrow represents
an operation execution.  It is known that set-linearizability has
strictly more expressiveness power than linearizability. Moreover, as
linearizability, set-linearizability is \emph{composable} (also called
\emph{local})~\cite{CRR18}.

\begin{figure}[ht]
  \begin{center}
    \vspace{0.7cm}
    \includegraphics[scale=0.7]{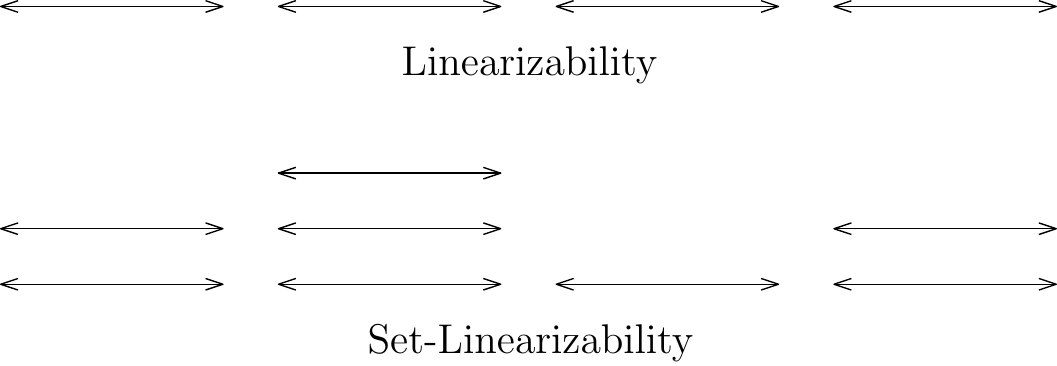}
    \caption{Graphical description of linearizability and
      set-linearizability.}
    \label{fig-example-linear}
  \end{center}
\end{figure}

A \emph{set-sequential specification} of a concurrent object differs
from a sequential execution in that $\delta$ receives as input the
current state $q$ of the machine and a set $Inv = \{ inv({\sf op_1}),
\hdots,inv({\sf op}_t) \}$ of operation invocations that happen
concurrently. Thus $\delta(q, Inv)$ returns $(q', Res)$ where $q'$ is
the next state and $Res = \{ res({\sf op_1}), \hdots, res({\sf op}_t)
\}$ are the responses to the invocations in $Inv$
(if the machine is non-deterministic, $Res$ is a set of sets of responses).
The sets $Inv$ and$Res$ are called \emph{concurrency classes}.
The sequences of invocation-response concurrency classes, $\langle INV:RES \rangle$, produced by the state machine are its
\emph{set-sequential executions}.
Observe that in a
sequential specification all concurrency classes have a single
element.

\begin{definition}[Set-linearizability]
Let $\mathcal A$ be an algorithm for a concurrent object $T$.
A finite execution $E$ of $\mathcal A$ is \emph{set-linearizable} if there is
a set-sequential execution $S$ of $T$ such that

\begin{enumerate}
  \item $S$ contains every completed operation of $E$ and might
    contain some pending operations.  Inputs and outputs of invocations
    and responses in $S$ agree with inputs and outputs in $E$.
  \item For every two completed operations {\sf op} and ${\sf op}'$ in
    $E$, if ${\sf op} <_E {\sf op}'$, then {\sf op} appears before ${\sf
      op}'$ in $S$.
\end{enumerate}

We say that $\mathcal A$
is \emph{set-linearizable} if each of its finite executions is set-linearizable.
\end{definition}

\section{Work-Stealing with Multiplicity}
\label{sec-ws-mult}

Work-stealing with \emph{multiplicity} is a relaxation of the usual
work-stealing in which, roughly speaking, every task is extracted
\emph{at least once}, and if it is extracted by several operations,
they must be \emph{concurrent}.
In the formal set-sequential specification below
(and in its variant in the next section),
tasks are inserted/extracted in FIFO order but it can be easily
adapted to encompass other orders (e.g. LIFO).
Figure~\ref{fig-example-execution} depicts an example of a
set-sequential execution of the work-stealing with multiplicity,
where concurrent \Take/\Steal operations can extract the same task.

\begin{figure}[ht]
  \begin{center} \vspace{0.4cm}
    \includegraphics[scale=0.52]{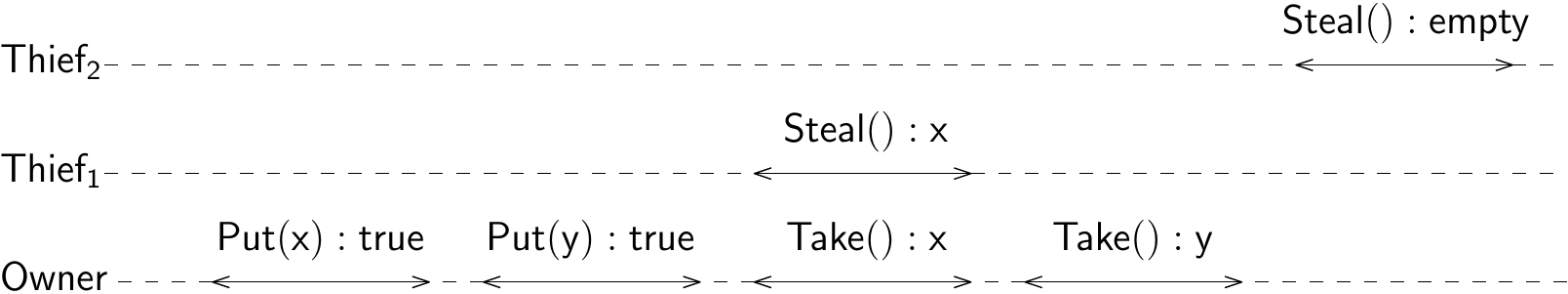}
    \caption{A set-sequential execution of work-stealing with
      multiplicity.}
    \label{fig-example-execution}
  \end{center}
\end{figure}

\begin{definition}[(FIFO) Work-Stealing with Multiplicity]
  \label{def-ws-mult}
  The universe of tasks that the owner can put is
  $\mathbf{N} = \{ 1, 2, \hdots \}$, and the set of states $Q$ is the
  infinite set of finite strings $\mathbf{N}^*$.  The initial state is
  the empty string, denoted~$\epsilon$.  In state $q$, the first element
  in $q$ represents the \emph{head} and the last one the
  \emph{tail}. The transitions are the following:

  \begin{enumerate}

    \item $\forall q \in Q\delta(q, \Put(x)) = (q \cdot x, \langle \Put(x): \true \rangle).$

    \item $\forall q \in Q$, $0\leq t\leq n-1$, $x \in \mathbf{N},\
      \delta(x \cdot q, \{\Take(), \Steal_1(), \hdots, \Steal_t() \})
      = \\ (q, \{ \langle \Take():x \rangle, \langle \Steal_1():x \rangle,
            \hdots, \langle \Steal_t():x \rangle \})$.

    \item $\forall q \in Q$, $1\leq t\leq n-1$, $x \in \mathbf{N},
      \delta(x \cdot q, \{\Steal_1(), \hdots, \Steal_t() \}) =
        (q, \{\langle \Steal_1():x \rangle, \hdots, \langle
        \Steal_t():x \rangle\}).$

    \item $\delta(\epsilon, \Take()) = (\epsilon, \langle \Take(): \epty
      \rangle).$

    \item $\delta(\epsilon, \Steal()) = (\epsilon, \langle \Steal(): \epty
      \rangle).$

  \end{enumerate}
\end{definition}

Let $\mathcal A$ be a linearizable algorithm for work-stealing with
multiplicity.  Note that items 2 and 3 in Definition~\ref{def-ws-mult}
and the definition of set-linearizability directly imply that in every
execution of $\mathcal A$, the number of \Take/\Steal operations that
take the same task is at most the number of processes in the system,
as the operations must be pairwise concurrent to be set-linearized
together.  Furthermore, every \emph{sequential} execution of $\mathcal
A$ looks like an ``exact'' solution for work-stealing, as every
operation is linearized alone, by definition of set-linearizability;
formally, every sequential execution of $\mathcal A$ is sequential
execution of (FIFO) work-stealing. We call this property
\emph{sequentially-exact}.  Thus, in the absence of contention,
$\mathcal A$ provides an exact solution for work-stealing.

\begin{remark}
\label{remark-seq-exact-set-lin}
Every set-linearizable algorithm for work-stealing with multiplicity
is sequentially-exact.
\end{remark}

\subsection{Work-Stealing with Multiplicity from \MaxReg}
\label{sec-ws-mult-max-reg}

Here we show that work-stealing with multiplicity can be reduced to
\MaxReg: a single wait-free linearizable \MaxReg object suffices to
wait-free set-linearizable solve work-stealing with multiplicity
without fences in all operations.  Roughly speaking, the \MaxReg
object synchronizes the head of the queue.  Therefore, there are fully
\R/\W algorithms for work-stealing with multiplicity as there are
\R/\W wait-free linearizable algorithms for \MaxReg.  We also argue
that using the \MaxReg algorithm in~\cite{AAC12}, the resulting
work-stealing with multiplicity algorithm has logarithmic step
complexity and do not use \RAW synchronization patterns in \Take and
\Steal.  Section~\ref{sec-ws-nc-mult} shows that the algorithm leads
to fully fence-free \R/\W work-stealing algorithm with constant step
complexity in all its operations.

Figure~\ref{figure-max-reg-mult} presents \WFWSM, a set-linearizable
algorithm for work-stealing with multiplicity. The algorithm is based
on a single wait-free linearizable \MaxReg object, which provides two
operations: \MaxR that returns the maximum value written so far in the
object, and \MaxW that writes a new value only if it is greater than
the largest value that has been written so far.

\begin{figure}[ht] \centering{
\fbox{
\begin{minipage}[t]{150mm} \small
\renewcommand{\baselinestretch}{2.5} \resetline
\begin{tabbing} aaaa\=aa\=aa\=aa\=aa\=aa\=aa\=\kill 

{\bf Shared Variables:}\\
\> $Head:$ atomic \MaxReg object initialized to $1$\\
\> $Tasks[1, 2, \hdots]:$ array of atomic \R/\W objects
with\\ \>\>\>\>\>\>\> the first two objects initialized to $\bot$\\ \\

{\bf Persistent Local Variables of the Owner:}\\
\> $tail \leftarrow 0$\\ \\

{\bf Operation} $\Put(x)$: \\
\line{A01} \> $tail \leftarrow tail+1$\\
\line{A02} \> $\{Tasks[tail].\W(x),\ Task[tail + 2].\W(\bot)\}$\\
\line{A03} \> {\bf return } {\sf true}\\
{\bf end} \Put \\ \\

{\bf Operation} $\Take()$: \\
\line{A04} \> $head \leftarrow Head.\MaxR()$\\
\line{A05} \> {\bf if $head \leq tail$ then}\\
\line{A06} \> \> $\{ x \leftarrow Tasks[head].\R(),
Head.\MaxW(head+1)\}$\\
\line{A09} \> \> {\bf return} $x$\\
\line{A10} \> {\bf end if}\\
\line{A11} \> {\bf return} \epty\\
{\bf end} \Take \\ \\

{\bf Operation} $\Steal()$: \\
\line{A12} \> $head \leftarrow Head.\MaxR()$\\
\line{A13} \> $x \leftarrow Tasks[head].\R()$ \\
\line{A14} \> {\bf if $x \neq \bot$ then}\\
\line{A16} \> \> $Head.\MaxW(head+1)$\\
\line{A17} \> \> {\bf return} $x$\\
\line{A18} \> {\bf end if}\\
\line{A19} \> {\bf return} \epty\\
{\bf end} \Steal

\end{tabbing}
\end{minipage}
}
\caption{\WFWSM: a \MaxReg-based set-linearizable
algorithm for work-stealing with multiplicity.}
\label{figure-max-reg-mult}
}
\end{figure}

In \WFWSM, the tail of the queue is stored in the local persistent
variable $tail$ of the owner, while the head is stored in the shared
\MaxReg $Head$.
When the owner wants to put a new task, it first locally increments
$tail$ (Line~\ref{A01}) and then stores the task in the corresponding
entry of $Tasks$ and marks one more entry with $\bot$
(Line~\ref{A02}); $\bot$ indicates lack of tasks.  Recall that the
notation in Line~\ref{A02} denotes that the instructions can be
executed in any order.
When the owner wants to take a task, it first reads the current head
of the queue from $Head$ (Line~\ref{A04}) and then, if there are tasks
available (i.e. the head is less or equal that the tail), it reads the
task at the head, updates $Head$ and finally returns the task
(Lines~\ref{A06} and~\ref{A09}); if there are no tasks available, the
owner returns \epty (Lines~\ref{A11}).
When a thief wants to steal a task, it first reads the current value
of $Head$ (Line~\ref{A12}) and then read that entry of $Tasks$
(Line~\ref{A13}).  If it reads a task (i.e. a non-$\bot$ value), it
updates $Head$ and then returns the task (Lines~\ref{A16}
and~\ref{A17}).  Otherwise, all tasks have been extracted and it
returns \epty (Line~\ref{A19}).

The semantics of \MaxW guarantees that $Head$ contains the
current value of the head of the queue at all times, as a ``slow''
process cannot ``move back'' the head by writing a smaller value in
$Head$ (in Lines~\ref{A06} or~\ref{A16}).  Thus, the \MaxReg $Head$
acts as a sort of barrier in the algorithm.  The net effect of this is
that the only way that two \Take/\Steal operations return the same
task is because they are concurrent, reading the same value from $Head$.

Note that if only the first object in $Tasks$ is initialized to $\bot$
(and hence \Put is modified accordingly), it is possible that a thief
reads an value from $Tasks$ that has not been written by the owner: in
an execution with a single $\Put(x)$ operation, the steps in
Line~\ref{A02} could be executed $Tasks[1].\W(x)$ first and then
$Task[2].\W(\bot)$ with a sequence of two \Steal operations completing
in between, hence the second operations reading $Tasks[2]$ which has
not been written yet by the owner, which is a problem if $Tasks[2]$
contains a value distinct from non-$\bot$ value.

\begin{theorem}
\label{theo-wf} Algorithm \WFWSM is a set-linearizable wait-free
fence-free algorithm for work-stealing with multiplicity using
atomic \R/\W objects and a single atomic \MaxReg object.
Moreover, all operations have constant step complexity
and \Put is fully \R/\W.
\end{theorem}

\begin{proof}
Clearly, the algorithm is wait-free as every
operation requires a constant number of steps to terminate,
hence having constant step complexity,
and $Head$ is wait-free, by assumption.
Observe that \Put uses only \R/\W.
Moreover, the algorithm  does not require any specific
ordering among its steps, beyond what is implied by data dependence,
therefore it is fully fence-free.

Before proving that \WFWSM is set-linearizible, we first observe
that at any time the thieves read the range of $Tasks$ that
the owner has already initialized; more specifically,
every \Steal operation reads from $Tasks$ (in Line~\ref{A13}),
a value that was written by the owner, either $\bot$ or a task.

At any time during an execution, the range
$Tasks[Head, Head+1, \hdots, ]$ contains a (possibly empty) sequence of
tasks followed by at least one $\bot$ value, considering the entries
in index-ascending order. The claim is true initially as the first
two entries if $Tasks$ are initialized to $\bot$.
Every time the owner stores a new task, it initializes a new entry of
$Task$ to $\bot$ (in Line~\ref{A02}); hence the claim holds at any time,
as $Head$ is incremented only if the owner or a thief sees that
$Tasks[Head]$ contains a non-$\bot$ value (in Lines~\ref{A06} or~\ref{A16}).
Note that the order of the instructions in Line~\ref{A02} is irrelevant.

We now prove that \WFWSM is set-linearizible. Consider any finite
execution $E$ of it. Since we already argued that the algorithm
is wait-free, there is a finite extension of $E$ in which all
its operations are completed and no new operations starts.
Thus, we can assume that there are no pending operations in $E$.

First, note that the semantics of \MaxW implies that there is no pair
of non-concurrent \Take/\Steal operation that return the same task: if
the two operation are not concurrent, then the first one increments
the value of $Head$ (and this value can only increment), and hence the
second cannot read the same tasks from $Tasks$. Thus, we have:

\begin{remark}
\label{remark-concurrency} If a task that is returned by more
than one \Take/\Steal operation, then
these operations are pairwise concurrent.  Thus, two distinct \Take
operations of the owner cannot return the same task.
\end{remark}

The main observation for the set-linearizability proof is that at any
time during the execution, the state of the object is represented by
the tasks in the range $Tasks[Head, Head+1, \hdots ]$,
i.e. the sequence of non-$\bot$ values (in index-ascending order)
that written by the owner in that range.
The set-linearization $\SetLin(E)$ of $E$ is obtained as follows:

\begin{itemize}

\item Every \Put operation is set-linearized \emph{alone}
(i.e. in a concurrency class containing only the operation) placed
at its step corresponding to  $Tasks[tail].\W(x)$ in $E$ (Line~\ref{A02}).

\item For every task that is returned by at least one \Take/\Steal
operation, all these operations are set-linearized in the same
concurrency class placed at the first step $e$ in $E$ that corresponds
to $Head.\MaxW(head+1)$ (either in Line~\ref{A06} or~\ref{A16}) among
the steps of the operations.
Note that that $e$ occurs between the invocation and response of every
operation in the concurrency class: since the operations return the
same task, all of them execute the \MaxR steps in Lines~\ref{A04}
or~\ref{A12} before $e$, and, by definition, $e$ appears in $E$ before
any other operation executes its step corresponding to $Head.\MaxW(head+1)$.
Observe that the order in which the instructions in Line~\ref{A06} is
executed is irrelevant.

\item Every \Take operation returning \epty is set-linearized
alone, placed at its step in $E$ corresponding to $Head.\MaxR()$
(Line~\ref{A04}).

\item Every \Steal operation returning \epty is set-linearized
alone, placed at its step in $E$ corresponding to $Head.\MaxR()$
(Line~\ref{A12}).

\end{itemize}

Every concurrency class of $\SetLin(E)$ is placed at a step of $E$
that lies between the invocation and response of each operation in the
concurrency class, which immediately implies that $\SetLin(E)$
respects the partial order $<_E$ of $E$. Thus, to conclude that
$\SetLin(E)$ is a set-linearization of $E$, we need to show that it is
indeed a set-sequential execution of the work-stealing with
multiplicity.

First, note that a task can be extracted by a \Take/\Steal operation
only if the \Put operation that stores the task executes its step
corresponding to $Tasks[tail].\W(x)$ (in Line~\ref{A02}) before
the \Take/\Steal operation reads the entry of $Tasks$
where the task is stored.  Thus, in $\SetLin(E)$ every task is
inserted before it is extracted.

Now, \Put stores tasks in $Tasks$ in index-ascending order.
Due to the semantics of \MaxReg, $Head$ never ``moves back'',
i.e. it only increments one by one, and hence
\Take and \Steal extract tasks in index-ascending order too.
It follows that tasks in $\SetLin(E)$ are inserted/extracted in FIFO order.
More specifically, for any
concurrency class $C$ of $\SetLin(E)$ with \Take/\Steal operations
that return the same tasks $x$, right before the step $e$ of $E$ where
$C$ is set-linearized, we have that $x$ is task with smallest index
(left-most) in the range $Tasks[Head, Head+1, \hdots]$, and thus indeed
the operations in $C$ get the oldest task in the object.

It only remains to be argued that any \Take/\Steal operation that
returns \epty, does so correctly, i.e., each of this operations are
set-linearized at a step of $E$ at which $Tasks[Head, Head+1, \hdots]$
is empty, i.e. all its entries initialized by the owner in that range
contain $\bot$.

Consider any \Take operation in $E$ that returns $\epty$.  Observe
that this can happen only if the owner sees that $head > tail$, namely
the conditional is not satisfied in Line~\ref{A05}.  Clearly, this is
possible only when no item has been inserted at all or when all items
have been extracted, and hence $Tasks[Head, Head+1, \hdots]$ is empty.
Consider now any \Steal operation in $E$ that returns $\epty$.  This
is possible only the thief reads $\bot$ from $Tasks$ in
Line~\ref{A13}, and since we already argued that the owner insert
tasks in ascending order, we have that at that time $Tasks[Head,
Head+1, \hdots]$ is empty.

We conclude that $\SetLin(E)$ is a valid set-sequential execution of
work-stealing with multiplicity, and as it respects the partial order
$<_E$ of $E$, we have that it is a set-linearization of $E$, and
therefore \WFWSM is set-linearizable. The theorem follows.
\end{proof}

When we replace $Head$ with
the wait-free linearizable \R/\W\ \MaxReg algorithm
in~\cite{AAC12}, whose step complexity is $O(\log m)$, where $m \geq 1$
is the maximum value that can be stored in the object, the step
complexity of \WFWSM is bounded wait-free with logarithmic step
complexity too. In the resulting algorithm at most $m$ tasks can be inserted.
Since the algorithm in~\cite{AAC12} do not use
\RAW synchronization patterns, the resulting algorithm does not uses
those patterns either.

\begin{theorem}
\label{theo-wf-log} If $Head$ is an instance of the wait-free
linearizable \R/\W\ \MaxReg algorithm in~\cite{AAC12},
\WFWSM is linearizable and
fully \R/\W with \Take and \Steal having step complexity
$O(\log m)$, where $m$ denotes the maximum number or tasks that can
be inserted in an execution. Furthermore, \Take and \Steal do not
use \RAW synchronization patterns.
\end{theorem}

\begin{proof}The algorithm remains linearizable, by
composability of linearizability~\cite{HW90}.
While the step complexity of \Put is clearly $O(1)$,
the step complexity of \Take and \Steal is $O(\log m)$
as the step complexity \MaxR and \MaxW of the \MaxReg algorithm
in ~\cite{AAC12} is $O(\log m)$.

We now argue that \Take and \Steal do not use \RAW synchronization
patterns.  The reason is that
the \MaxReg algorithm in~\cite{AAC12} does not use this
synchronization mechanism.  Roughly speaking, the algorithm consists
of a binary tree of height $O(\log m)$ with an atomic bit in each of
its nodes.  When a process wants to perform \MaxR, it reads the bits
in a path of the tree from the root to a leaf and then returns a
value, according to the leaf it reached; the next node in the path the
process reads depends on the value of the current node.  When a
process wants to perform \MaxW, it reads the bits in a path from the
root to a leaf, which is on function of the binary representation of
the value the process wants to write; then, if the new value is larger
than the current one, in a bottom-up manner, it writes 1 in every node
in the path with 0 (for the algorithm to be linearizable, the writes
should occur in this order).  Thus, we have that \MaxR consists of a
sequence of reads and \MaxW consists of a sequence of reads followed
by a (possibly empty) sequence of writes.  Therefore, \Take/\Steal of
\WFWSM consists of a sequence of reads followed by a (possibly empty)
sequence of writes, and thus the operation does not use \RAW
synchronization patterns.
\end{proof}

\section{Work-Stealing with Weak Multiplicity}
\label{sec-ws-nc-mult}

In the context of work-stealing, a logarithmic step complexity of the
\Take operation may be prohibitive in practice. Ideally, we would like to
have constant step complexity, in all operations if possible, and using
simple synchronization mechanisms of course.
In this section, we propose a variant of work-stealing with
multiplicity, which admits fully \R/\W fence-free implementations
with constant step complexity in all its operations.
Intuitively, the variant requires that every task is extracted at
least once, but now every process extracts a task \emph{at most once},
hence \Take/\Steal operations returning the same task \emph{may not}
be concurrent.  Therefore, the relaxation retains the property that
the number of operations that can extract the same task is at most the
number of processes in the system.  We call this relaxation
work-stealing with \emph{weak} multiplicity.

\begin{figure}[ht]
\begin{center}
\includegraphics[scale=0.6]{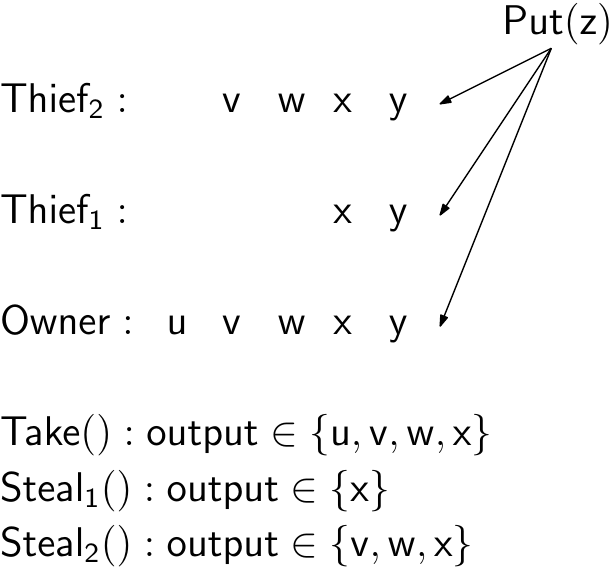}
\caption{A schematic view of work-stealing with weak multiplicity.}
\label{fig-state-weak}
\end{center}
\end{figure}

Figure~\ref{fig-state-weak} depicts a schematic view
of a state of work-stealing with weak multiplicity.
Intuitively, at any state, each process has its own queue of tasks.
When the owner inserts a new task, it concurrently places the task
in all queues, as shown in the figure.
Therefore, in any state, for any pair of process' queues,
one of them is suffix of the other.
A \Take/\Steal operation can return any task from its
queue that is no ``beyond'' the first task of the \emph{shortest} queue
in the state. In the example, a \Steal operation of thief $p_1$,
denoted $\Steal_1()$, can return only $x$, as $p_1$ has the shortest
queue in the state, while a \Take operation of the owner can return
any task in $\{u,v,w,x\}$, which contains any task from the beginning
of its queue up to $x$.

\begin{figure}[ht]
\begin{center}
\includegraphics[scale=0.45]{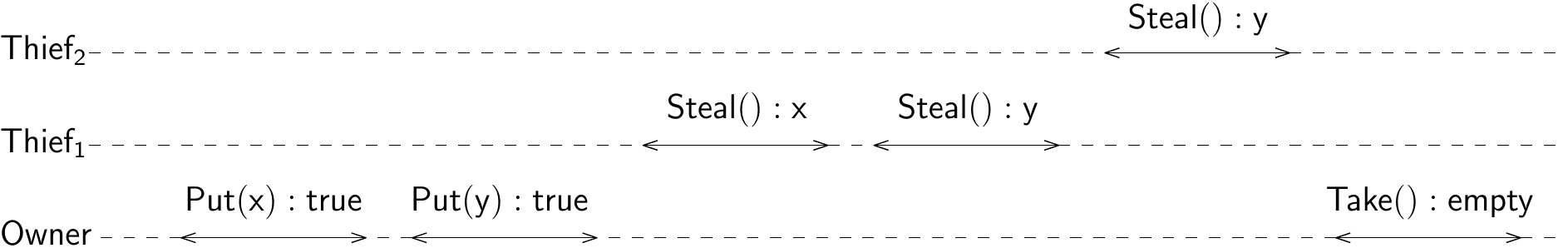}
\caption{A sequential execution of work-stealing with
weak multiplicity.}
\label{fig-example-execution-weak}
\end{center}
\end{figure}

Figure~\ref{fig-example-execution-weak} shows an
example of a sequential execution of work-stealing with
weak multiplicity; note that \Take/\Steal operations
are allowed to get the same item, although they are not concurrent.

The next \emph{sequential} specification formally defines
work-stealing with weak multiplicity.
{Without loss of generality, the specification assumes, that $p_0$
is the owner, and each invocation/response of
thief $p_i$ is subscripted with its index,
which then belongs to $\{1, \hdots, n-1\}$.}

\begin{definition}[(FIFO) Work-Stealing with Weak
Multiplicity]
  \label{def-ws-nc-mult} The universe of tasks that the owner can put
  is $\mathbf{N} = \{ 1, 2, \hdots \}$, and the set of states $Q$ is the
  infinite set of $n$-vectors of finite strings $\mathbf{N}^* \times
  \hdots \times \mathbf{N}^*$, with the property that for any two pairs
  of strings in a vector, one of them is a \emph{suffix of the other}.
  The initial state is the vector with empty strings, $(\epsilon, \hdots, \epsilon)$.
  The transitions are the following:
  \begin{enumerate}

    \item $\forall (t_0, \hdots, t_{n-1}) \in Q,         \delta((t_0,
      \hdots, t_{n-1}), \Put(x)) =
        ((t_0 \cdot x, \hdots, t_{n-1} \cdot x),
        \langle \Put(x): \true \rangle).$

    \item {$\forall (t_0, \hdots, t_{n-1}) \in Q$ such that
    $t_0 = x_1 \cdots x_j \cdot q \neq \epsilon$ with $j \geq 1$
    and $x_j \cdot q$ being the shortest string in the state
    (possibly with $x_j \cdot q = \epsilon$), $        \delta((t_0,
    \hdots, t_{n-1}), \Take()) =  \{ ((\widehat t_0, t_1, \hdots,
    t_{n-1}), \langle \Take():x_k \rangle) \}$, where $k \in \{1,
    \hdots, j\}$ and $\widehat t_0 = x_{k+1} \cdots x_j \cdot q$.}

    \item {$\forall (t_0, \hdots, t_{n-1}) \in Q$ such that
    $t_i = x_1 \cdots x_j \cdot q \neq \epsilon$ with $j \geq 1$,
    $i \in \{1, \hdots, n-1\}$ and $x_j \cdot q$ being the shortest
    string in the state (possibly with $x_j \cdot q = \epsilon$), $
    \delta((t_0, \hdots, t_{n-1}), \Steal_i()) =
    \{ ((t_0, \hdots, t_{i-1}, \widehat t_i, t_{i+1}, \hdots, t_{n-1}),
    \langle \Steal_i():x_k \rangle) \}$,
    where $k \in \{1, \hdots, j\}$ and $\widehat{t_i} = x_{k+1} \cdots
    x_j \cdot q$.}

    \item {$\forall (\epsilon, t_1, \hdots, t_{n-1}) \in Q$, $
        \delta((\epsilon, t_1, \hdots, t_{n-1}), \Take()) =
        ((\epsilon, t_1, \hdots, t_{n-1}), \langle \Take():\epsilon
        \rangle).$}

    \item {$\forall (t_0, \hdots, t_{i-1}, \epsilon, t_{i+1}, \hdots,
        t_{n-1}) \in Q$ with $i \in \{1, \hdots, n-1\}$,$
        \delta((t_0, \hdots, t_{i-1}, \epsilon, t_{i+1}, \hdots,
        t_{n-1}), \Steal_i()) = \\((t_0, \hdots, t_{i-1}, \epsilon,
        t_{i+1}, \hdots, t_{n-1}), \langle \Steal_i():\epsilon \rangle).$}
\end{enumerate}
\end{definition}

Observe that the second and third items in Definition~\ref{def-ws-nc-mult}
correspond to non-deterministic transitions in which a \Take/\Steal
operation can extract any task of $\{x_1, \hdots, x_j\}$;
the value returned can be $\epsilon$ when the shortest string in the state
($x_j \cdot q'$ in the definition) is $\epsilon$.
Furthermore, the definition guarantees that every task is extracted
at least once because every \Take/\Steal operation can only return a task
that is not ``beyond'' the first task in the shortest string of the state.

The specification work-stealing with concurrent multiplicity
is nearly trivial,
with solutions requiring almost no synchronization.
A simple solution is obtained by replacing $Head$ in \WFWSM
with one local persistent variable $head$ per process
(each initialized to 1); \Put remains the same and \Take and \Steal instead
of reading from $Head$, they just locally read the current value of $head$
and increment it whenever a task is taken. It is not hard to verify
that the resulting algorithm is indeed linearizable.

To avoid this kind of simple solutions (which would very inefficient in
practice as every process is processing every task),
we restrict our attention to sequentially-exact algorithms,
namely, every sequential execution of the algorithm is a sequential execution
of the specification of (FIFO) work-stealing.\footnote{Alternatively,
work-stealing with weak multiplicity can be specified using
the interval-linearizability formalism in~\cite{CRR18}, which allows us
to specify that a \Take/\Steal operation can exhibit a non-exact behavior
only in presence of concurrency. Interval-linarizability would directly
imply that any interval-linearizable solution provides a exact solution
in sequential executions.
Roughly speaking, in interval-linearizability,
operations are linearized at intervals that can overlap each other.}
It is easy to see that the algorithm described above does not have this
property.

Finally, we stress that, differently from work-stealing with multiplicity,
two distinct non-concurrent
\Take/\Steal can extract the same task in an execution, which can happens
only if \emph{some} operations are concurrent in the execution,
due to the sequentially-exact requirement.
Particularly, in our algorithms, this relaxed behaviour can occurs
when processes are concurrently updating the head of the queue.

\subsection{Fully \R/\W Fence-Free Work-Stealing with Multiplicity}
\label{sec-ws-mult-read-write}

In this subsection we present \NCWSM, a fully \R/\W fence-free
algorithm for work-stealing with weak multiplicity.
The algorithm is obtained by replacing the atomic \MaxReg object in
\WFWSM, $Head$, with an atomic \RangeMaxReg object, a relaxation of \MaxReg
with a \RMaxR operation that returns a value in a \emph{range} of values
that have been written in the register; the range always includes
the maximum value written so far.

We present a \RangeMaxReg algorithm
that is nearly trivial, however, it allows us to solve
work-stealing with weak multiplicity in an efficient manner,
with implementations exhibiting good performance
in practice, as we will see in Section~\ref{sec-experiments}.
As above, to avoid trivial solutions,
we focus on sequentially-exact linearizable algorithms for \RangeMaxReg
with each of its sequential executions being a sequential
execution of \MaxReg.~\footnote{Again, this can be alternatively
specified through interval-lineiarizability.}
When \NCWSM is combined with this algorithm, it becomes
fully \R/\W, fence-free, linearizable, sequentially-exact and
wait-free with constant step complexity.

Intuitively, in \RangeMaxReg, each process has a \emph{private}
\MaxReg and whenever it invokes \RMaxR, the result lies
in the range defined by the value of its private \MaxReg and
the maximum among the values of the private {\sf MaxRegisters} in the state.
In the sequential specification of \RangeMaxReg, each invocation/response
of process $p_i$ is subscripted with its index $i$.

\begin{definition}[\RangeMaxReg]
  \label{def-range-max-reg} The set of states $Q$ is the infinite set of
  $n$-vectors with natural numbers, with vector $(1, \hdots, 1)$ being
  the initial state. The transitions are the following:
  \begin{enumerate}

    \item $\forall (r_0, \hdots, r_{n-1}) \in Q$ and $i \in \{0, \hdots, n-1\}$,
      if $x > r_i$,\\
      $\delta((r_0, \hdots, r_{n-1}),\ \RMaxW_i(x))
      = ((r_0, \hdots, r_{i-1}, x, r_{i+1}, \hdots, r_{n-1}), \langle
      \RMaxW_i(x): \true \rangle),$\\
    otherwise\\
    $\delta((r_0, \hdots, r_{n-1}), \RMaxW(x)) =
    ((r_0, \hdots, r_{n-1}), \langle \RMaxW(x): \true \rangle).$

    \item $\forall (r_0, \hdots, r_{n-1}) \in Q$ and $i \in \{0,
      \hdots, n-1\}$,
      $\delta((r_0, \hdots, r_{n-1}), \RMaxR_i()) = \\
        \{((r_0, \hdots, r_{i-1}, x, r_{i+1}, \hdots, r_{n-1}),
        \langle \RMaxR_i(): x \rangle)\},$
        where $x \in \{r_i, r_i+1, \hdots, \max(r_0, \hdots, r_{n-1})\}$.
\end{enumerate}
\end{definition}

As already mentioned, \NCWSM is the algorithm obtained by replacing
$Head$ in \WFWSM with an atomic \RangeMaxReg object
initialized to 1 (hence \MaxR and \MaxW are replaced by \RMaxR and \RMaxW,
respectively).

\begin{theorem}
\label{theo-wf-nc}
Algorithm \NCWSM is a linearizable wait-free
fence-free algorithm for work-stealing with weak multiplicity
using atomic \R/\W objects and a single atomic \RangeMaxReg object.
Moreover, all operations have constant step complexity
and \Put is fully \R/\W.
\end{theorem}

\begin{proof} Clearly, all operations of the algorithm are
wait-free since each of them requires a constant number of steps to
terminate, hence having constant step complexity,
and $Head$ is wait-free, by assumption.
Note that \Put uses only \R/\W.
The algorithm does not require any specific
ordering among its steps, beyond what is implied by data dependence,
therefore it is fully fence-free.

As in the proof of Theorem~\ref{theo-wf}, it can be argued
that thieves read values from $Tasks$ that have been written
by the owner.

Consider any finite execution $E$ of the algorithm.
Since the algorithm is
wait-free, there is a finite extension of $E$ in which all its
operations are completed and no new operations starts. Thus, we can
assume that there are no pending operations in $E$.

Proving that $E$ is linearizable is quite straightforward. It is
enough to observe that, at any step of $E$, the state of the object
is encoded in the state of $Head$. Let $(r_0, \hdots, r_{n-1})$
be the state of $Head$ at a given step of $E$. Then,
the state of the object is $(t_0, \hdots, t_{n-1})$ with
each $t_i$ being the finite sequence of tasks in the range
$Tasks[r_i, r_{i+1}, \hdots]$ (i.e. the sequence of non-$\bot$ values
written by the owner, in index-ascending order).
Thus, in a linearization of $E$, a $\Put(x)$ operation is linearized at
its step $Tasks[tail].\W(x)$ in Line~\ref{A02},
while a \Take/\Steal operation is linearized at its
step $Head.\RMaxR()$ in Line~\ref{A04}/Line~\ref{A12}
(observe that the non-deterministic choice in a transition with a
\Take/\Steal operation is resolved with the outcome of \RMaxR).
Therefore, we conclude that every execution of \NCWSM is linearizable,
and thus the algorithm is linearizable too.
\end{proof}


\begin{figure}[ht] \centering{ \fbox{
\begin{minipage}[t]{150mm} \small
\renewcommand{\baselinestretch}{2.5} \resetline
\begin{tabbing} aaaa\=aa\=aa\=aa\=aa\=aa\=aa\=\kill 

{\bf Shared Variables:}\\
\> $R:$ atomic \R/\W object initialized to 1\\ \\

{\bf Persistent Local Variables of a Process:}\\
\> $r \leftarrow 1$\\ \\

{\bf Operation} $\RMaxW(x)$: \\
\line{C01} \> $r \leftarrow \max\{r, R.\R()\}$ \\
\line{C02} \> {\bf if $x > r$ then} \\
\line{C03} \> \> $\{r \leftarrow x, R.\W(x)\}$\\
\line{C04} \> {\bf end if}\\
\line{C05} \> {\bf return } {\sf true}\\
{\bf end} \RMaxW \\ \\

{\bf Operation} $\RMaxR()$: \\
\line{C06} \> $r \leftarrow \max\{r, R.\R()\}$\\
\line{C07} \> {\bf return} $r$\\
{\bf end} \RMaxR

\end{tabbing}
\end{minipage} }
\caption{A linearizable wait-free algorithm for \RangeMaxReg.}
\label{figure-algo-range-max-reg}}
\end{figure}


Figure~\ref{figure-algo-range-max-reg} contains a simple linearizable
sequentially-exact wait-free algorithm for \RangeMaxReg.
All processes share a single \R/\W object~$R$
and each process has a local persistent variable~$r$.
The idea is very simple: each process locally stores in $r$
the maximum value it is aware of;
whenever it discover a new largest value in \RMaxW,
it writes it in $r$ and $R$ and since $R$ might not have the largest value,
it returns the maximum among $r$ and $R$ in \RMaxR.

\begin{theorem}
\label{theo-range-max-reg}
The algorithm in Figure~\ref{figure-algo-range-max-reg} is a linearizable
sequen-tially-exact wait-free fence-free algorithm for
\RangeMaxReg using only atomic \R/\W objects and with
constant step complexity in all its operations.
\end{theorem}

\begin{proof}
It is clear form its pseudocode that the algorithm is
fully \R/\W, wait-free and fence-free,
and each operations has constant step complexity.

Consider any finite execution $E$ of the algorithm.
Since the algorithm is
wait-free, there is a finite extension of $E$ in which all its
operations are completed and no new operations starts. Thus, we can
assume that there are no pending operations in $E$.

To prove linearizability, it is enough to observe
that at any step of $E$, the state of the object
is $(r_0, \hdots, r_{n-1})$, where $r_i$ is the value stored in the
local persistent variable $r$ of process $p_i$ at that moment.
Thus, a $\RMaxW(x)$ operation with $x > r$ (hence the conditional
in Line~\ref{B01} is {\sf true}) is linearized at its step
$R.\W(x)$ in Line~\ref{B02}; if $x \leq r$, the operation is linearized
at is invocation.
A $\RMaxR()$ operation is linearizaed at its step $R.\R()$ in Line~\ref{B05};
note that the operation returns a value between the value in $r$
before Line~\ref{B05}
and the maximum among the $r$'s local variables, since that is the maximum
value $R$ can store at that time. Thus the algorithm is linearizable.

Suppose now that $E$ is sequential. By induction of the number of operations,
it is easy to show that $R$ contains the maximum value, at all time.
Thus, $E$ is indeed a sequential execution of \MaxReg,
and therefore the algorithm is sequentially-exact. The theorem follows.
\end{proof}

We now are able to present our main result:

\begin{theorem}
\label{theo-wf-fully} If $Head$ is an instance of the
algorithm in Figure~\ref{figure-algo-range-max-reg},
\NCWSM is fully \R/\W, fence-free, wait-free, sequentially-exact
and linearizable with constant step complexity in all its operations.
\end{theorem}

\begin{proof} By composability of linearizability~\cite{HW90},
the algorithm remains linearizable.
The algorithm is fully \R/\W and wait-free because \Put uses only \R/\WFWSM
and the \RangeMaxReg algorithm in Figure~\ref{figure-algo-range-max-reg}
is fully \R/\W and wait-free, by Theorem~\ref{theo-range-max-reg}.
Clearly, the step complexity of \Put is $O(1)$.
The step complexity of \Take and \Steal is $O(1)$ too
because the \RangeMaxReg algorithm in Figure~\ref{figure-algo-range-max-reg}
has constant step complexity.
The resulting algorithm does not require any specific
ordering among its steps, beyond what is implied by data dependence,
therefore it is fully fence-free.
Consider a sequential execution of the algorithm.
By Theorem~\ref{theo-range-max-reg}, the algorithm
in Figure~\ref{figure-algo-range-max-reg} behaves like a \MaxReg,
and hence the execution corresponds to a sequential execution of \WFWSM
(exchanging \RMaxR and \RMaxW with \MaxR and \MaxW, respectively),
which in turn is a sequential execution of work-stealing,
since \WFWSM is sequentially-exact, by Remark~\ref{remark-seq-exact-set-lin}
and Theorem~\ref{theo-wf}.
Thus, the algorithm is sequentially exact. The theorem follows.
\end{proof}

\begin{figure}[ht] \centering{ \fbox{
\begin{minipage}[t]{150mm} \small
\renewcommand{\baselinestretch}{2.5} \resetline
\begin{tabbing} aaaa\=aa\=aa\=aa\=aa\=aa\=aa\=\kill 

{\bf Shared Variables:}\\
\> $Head:$ atomic \R/\W object initialized to 1\\
\> $Tasks[1, 2, \hdots]:$ array of atomic \R/\W objects \\
\> \> \> \> \> \> with the first two objects initialized to $\bot$\\ \\

{\bf Persistent Local Variables of the Owner:}\\
\> $head \leftarrow 1$\\
\> $tail \leftarrow 0$\\ \\

{\bf Persistent Local Variables of a Thief:}\\
\> $head \leftarrow 1$\\ \\

{\bf Operation} $\Put(x)$: \\
\line{B01} \> $tail \leftarrow tail+1$\\
\line{B02} \> $\{Tasks[tail].\W(x),\ Tasks[tail + 2].\W(\bot)\}$\\
\line{B03} \> {\bf return } {\sf true}\\
{\bf end} \Put \\ \\

{\bf Operation} $\Take()$: \\
\line{B04} \> $head \leftarrow \max\{head, Head.\R()\}$\\
\line{B05} \> {\bf if $head \leq tail$ then}\\
\line{B06} \> \> $\{ x \leftarrow Tasks[head].\R(),Head.\W(head+1)\}$\\
\line{B07} \> \> $head \leftarrow head+1$\\
\line{B08} \> \> {\bf return} $x$\\
\line{B09} \> {\bf end if}\\
\line{B10} \> {\bf return} \epty\\
{\bf end} \Take \\ \\

{\bf Operation} $\Steal()$: \\
\line{B11} \> $head \leftarrow \max\{head, Head.\R()\}$\\
\line{B12} \> $x \leftarrow Tasks[head].\R()$ \\
\line{B13} \> {\bf if $x \neq \bot$ then}\\
\line{B15} \> \> $Head.\W(head+1)$\\
\line{B14} \> \> $head \leftarrow head+1$\\
\line{B16} \> \> {\bf return} $x$\\
\line{B17} \> {\bf end if}\\
\line{B18} \> {\bf return} \epty\\
{\bf end} \Steal

\end{tabbing}
\end{minipage} }
\caption{\NCWSM algorithm with the \RangeMaxReg
algorithm in Figure~\ref{figure-algo-range-max-reg} inlined.}
\label{figure-w-mult} }
\end{figure}

Figure~\ref{figure-w-mult} contains an optimized \NCWSM
algorithm with the \RangeMaxReg algorithm in
Figure~\ref{figure-algo-range-max-reg} inlined.
Since \Take and \Steal first \RMaxR from and then \RMaxW to $Tail$,
the algorithm remains sequentially-exact when removing Line~\ref{C01}
of \RMaxW in Figure~\ref{figure-algo-range-max-reg}.
Our experimental evaluation in Section~\ref{sec-experiments}
tested implementations of this algorithm.

\section{Bounding the Multiplicity}
\label{sec-bound-mult}

Here we discuss simple variants of our algorithms that bound the
number of operations that can extract the same task.  More in detail,
with the addition of a single \SWAP instruction in \Steal, the
modified algorithms guarantee that no two distinct \Steal operations
take the same task (however, a \Take and a \Steal can take the same
task).

We only discuss the case of \WFWSM as the variant for \NCWSM is the
same. The modification consists in having an array $A$ of the same
length of $Tasks$, with each entry initialized to \true. After
Line~\ref{A14}, a thief performs $A[head].\SWAP(\false)$, and it
executes Lines~\ref{A16} and~\ref{A17} only if the \SWAP successfully
takes the \true value in $A[head]$; otherwise, it increments $head$
and goes to Line~\ref{A12} to start over.  Therefore, \Steal is only
nonblocking in the modified algorithm. This \emph{bounded} variant of
\WFWSM is denoted \BNBWSM.

The set-linearization proof of \BNBWSM is almost the same,
with the difference that for every task that is stored in $Tasks[r]$
and is returned by one or two \Take/\Steal operations, the operations
are set-linearized in the same concurrency class placed at the first
step $e$ in $E$ that executes $Head.\MaxW(head+1)$, in case of \Take
(as in the proof of Theorem~\ref{theo-wf}),
or $A[r].\SWAP(\top)$, in case of \Steal.

\paragraph*{Removing Multiplicity}
The \Take operation of \WFWSM can be modified similarly to obtain
solutions for exact (FIFO) work-stealing.  The modification consists
in using the array $A$ mentioned above. After Line~\ref{A05}, the
owner performs $A[head].\SWAP(\false)$, and it executes
Lines~\ref{A06} and~\ref{A09} only if the \SWAP successfully takes the
\true value in $A[head]$; otherwise, it increments $head$ and goes to
Line~\ref{A04} to start over.

\section{Implementing Arrays of Infinite Length}
\label{sec-removing-infinite-arrays}

So far we have presented our algorithms using an infinite length array
where the tasks are stored. In this section we discuss two approaches
to implement our algorithms using arrays of finite length;
both approaches have been used in previous algorithms
(e.g.~\cite{AF20, AKY10, MlVS09, HLMS06, YM16}).
We only discuss the case of \WFWSM as the other cases are
handled in the same way.
The two approaches are the following:

\begin{enumerate}
  \item\label{desc-item1}
  In the first approach, the algorithm starts with $Tasks$ pointing to
  an array of finite fixed length, with its two first objects
  initialized to $\bot$; each time the owner detects the array is full
  (i.e. when $tail$ is larger the length $Tasks$), in the middle of a
  \Put operation, it creates a new array $A$, duplicating the previous
  length, copies the previous content to $A$, initializes the next
  two objects to $\bot$, points $Tasks$ to~$A$ and finally continues
  executing the algorithm.
  Although the modified \Put operation remains wait-free,
  its step complexity is unbounded.

  \item\label{desc-item2}
  In the second approach, $Tasks$ is implemented with a
  linked list with each node having a fixed length array.
  Initially, $Tasks$ consists of a single node, with the first two
  objects of its array initialized to $\bot$.
  When the owner detects that all entries in the linked list have been
  used, in the middle of a \Put operation, it creates a new node,
  initializes the first two objects to $\bot$, links the new node to
  the end of the list and continues executing the algorithm.
  An index of $Tasks$ is now made of a tuple:
  a pointer to a node of the linked list and an node-index array.
  Thus, any pair of nodes can be easy compared (first pointer nodes,
  then node-indexes) and incrementing an index can be easily
  performed too (if the node-index is the last one,
  the pointer moves forward and the node-index is set to one,
  otherwise only the node-index is incremented).
  The modified \Put operation remains wait-free with
  constant step complexity.
\end{enumerate}

Our experimental evaluation in Section~\ref{sec-experiments} shows
the second approach performs better than the first one when solving
a problem of concurrent nature.
\section{Idempotent \texorpdfstring{$\neq$}{≠} Multiplicity}
\label{sec-idem-neq-mult}

Idempotent work-stealing is (only) informally defined in~\cite{MlVS09}
as: every task is extracted \emph{at least once}, instead of
\emph{exactly once} (in some order).  Three idempotent work-stealing
algorithms are presented in~\cite{MlVS09}, inserting/extracting tasks
in FIFO and LIFO order, and as a double-ended queue (the owner puts in
and takes from one side and the thieves steal from the other side).

\begin{figure}[ht]
  \begin{center} \vspace{0.55cm}
    \includegraphics[scale=0.5]{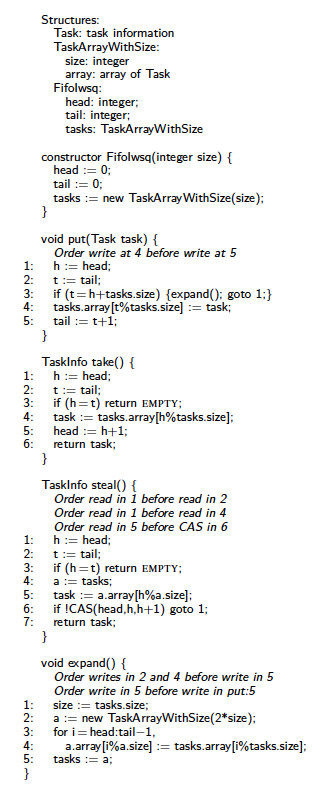}
    \caption{\small Idempotent FIFO work-stealing~\cite{MlVS09}.}
    \label{fig-idempotent-fifo}
  \end{center}
\end{figure}

Here we explain that these algorithms does not implement work-stealing
with multiplicity, neither its non-concurrent variant.  While in our
relaxations every process extracts a task at most once, and hence the
number of distinct operations that extract the same task is at most
the number of processes in the system, in idempotent work-stealing a
task can be extracted an unbounded number of times, and moreover, a
thief can extract the same task an unbounded number of times.

Figure~\ref{fig-idempotent-fifo} depicts the FIFO idempotent
work-stealing algorithm in~\cite{MlVS09}.  The algorithm stores the
tasks in a shared array $tasks$ and shared integers $head$ and $tail$
indicate the positions of the head and the tail.  For every integer
$z > 0$, we describe an execution of the algorithm in which, for every
$k \in \{1, \hdots, z\}$, there is a task that is extracted by
$\Theta(k)$ distinct operations (possibly by the same thief), with
only one of them being concurrent with the others.

\begin{enumerate}

\item Let the owner execute alone $z$ times \Put. Thus, there are $z$
distinct tasks in $tasks$

\item Let $r = z$.

\item The owner executes \Take and stops before executing Line~5, i.e.
it is about to increment $head$.

\item In some order, the thieves sequentially execute $r$ \Steal
operations; note these \Steal operations return the $r$ tasks in
$tasks[0, \hdots, r-1]$.

\item We now let the owner increment $head$. If $r > 1$, go to step 3
with $r$ decremented by one, else, end the execution.
\end{enumerate}

Observe that in the execution just described, the task in $tasks[i]$,
$i \in \{0, \hdots, z-1\}$, is extracted by a \Take operation and by
$i+1$ distinct non-concurrent \Steal operations (possible by the same
thief).  Thus, the task is extracted $\Theta(i)$ distinct times.
Since $z$ is any positive integer, we conclude that there is no bound
on the number of times a task can be extracted.

A similar argument works for the other two idempotent work-stealing
algorithms in~\cite{MlVS09}.  In the end, this happens in all
algorithms because tasks are not marked as taken in the shared array
where they are stored. Thus, when the owner takes a task and
experience a delay before updating the head/tail, all concurrent
modifications of the head/tail performed by the thieves are
overwritten once the owner completes its operation, hence leaving all
taken tasks ready to be taken again.

\section{Experiments}\label{sec-experiments}

In this section we discuss the outcome of the experiments we have
conducted to evaluate the performance of implementations of \NCWSM and
its bounded version, \BNCWSM (see Section~\ref{sec-bound-mult}).
Based on the approaches discussed in
Section~\ref{sec-removing-infinite-arrays}, we implemented two
versions of the algorithm, one using arrays and another using linked
lists.
We mainly discuss the results of the version based on linked
lists, since it exhibited better performance than the version based on
arrays.
\NCWSM and \BNCWSM were compared to the following algorithms:
THE Cilk~\cite{CL05}, Chase-Lev~\cite{FLR98}, and the three Idempotent
Work-Stealing algorithms~\cite{MlVS09}.

\subsection{Platform and  Environment}
\label{subsec-platform-and-environment}

The experiments were executed in two machines with distinct
characteristics. The first machine has an Intel Kaby Lake
processor (i7–7700HQ, with four cores where each core has two threads
of execution) and 16GB RAM. This machine was dedicated only to the
experiments. The second one is a cluster with four Intel Xeon
processors (E7–8870 v3 processor, with 18 cores where each core has
two threads of execution) and 3TB RAM. This machine was shared with
other users, where the process executions are closer to a usual
situation, namely, resources (CPU and memory) are shared by all
user. We will focus mainly on the results of the experiments in the
Core i7 processor, and additionally will use the results of the
experiments in Intel Xeon to complement our analysis.  The algorithms
were implemented in the Java platform (OpenJDK 1.8.0\_275) in order to
test them in a cross-platform computing environment.

\subsection{Methodology}\label{subsec:methodology}

To analyze the performance of the algorithms, we divided the analysis
into the next two benchmarks, which have been used
in~\cite{FLR98,MlVS09,MA14}:

\begin{itemize}
  \item \textbf{Zero cost experiments}. We test the performance
    of executing only the methods provided by the algorithms. Thus, we
    compare the time taken by each algorithm when performing the same set
    of tasks without being used as support for other algorithms.
    In our experiments, we test the time required for executing
    by \Put-\Take and \Put-\Steal operations
  \item \textbf{Irregular graph application}. We test the
    performance of solving a challenging graph problem, computing a
    spanning tree, where the goal is to speed up the computation via
    parallel exploration of the graph.
\end{itemize}

Below, we explain in detail how we implement each benchmark.

\paragraph{Zero cost experiments} We measure the time required for
performing a sequence of operations provided by the work-stealing
algorithms. We measure the time needed for \Put-\Take operations,
where the owner performs a sequence of \Put operations followed by an
equal number of \Takes.
Differently from~\cite{MlVS09, MA14}, we also measure the time for
\Put-\Steal operations.
In both experiments, The number of \Put operations is $10,000,000$,
followed by the same number of \Take or \Steal operations;
no operation performs any work associated to a task.

\paragraph{Irregular graph application}

We consider the spanning-tree problem to evaluate the performance of
each algorithm. It is measured the speed up of the computation by the
parallel exploration of the input graph. This problem is used
in~\cite{MlVS09, MA14} to evaluate their work-stealing algorithms.  We
refer the reader to~\cite{BaderC04} for a detailed description of the
algorithm.
The spanning tree algorithm already uses a form of work-stealing to
ensure load-balancing, and we adapted it to work with our
work-stealing implementations. Further, the algorithms were tested on
many types of directed and undirected graphs:

\begin{itemize}
  \item \textbf{Regular and irregular meshes}
  \begin{itemize}
    \item \textbf{2D Torus}: The vertices are on a 2D mesh, where each
      vertex has a connection to its four neighbors in the mesh.
    \item \textbf{2D60 Torus}: It is a random graph obtained from the
      previous one, where each edge has a probability of 60\% to be
      present.
    \item \textbf{3D Torus}: The vertices are on a 3D mesh, where each
      vertex has a connection to its six neighbors in the mesh.
    \item \textbf{3D40 Torus}: It is a random graph obtained from the
      previous one where each edge has a probability of 40\% to be
      present.
  \end{itemize}
  \item \textbf{Random}: Random graph of \(n\) vertices and \(m\) edges,
    by randomly adding \(m\) unique edges to the vertex set.
\end{itemize}

All graphs are represented using the adjacency lists
representation. The graphs were built with 1,000,000 of vertices for
tests in core i7 and 2,000,000 for tests in Intel Xeon.

Each experiment consists of the following: take one of the previous
graphs, specify its parameters and run the spanning-tree algorithm
with one of the of the work-stealing algorithms (THE Cilk, Chase-Lev,
Idempotent FIFO, Idempotent LIFO, Idempotent DEQUE, \NCWSM or
\BNCWSM); the spanning tree algorithm is executed five times with each
work-stealing algorithm, using the same root and the same graph.  In
each execution, it is tested the performance obtained by increasing
the number of threads from one to the total supported by the
processor, registering the duration of every execution.

For each work-stealing algorithm and number of threads, the fastest
and slowest executions are discarded, and the average of the remaining
three is calculated. All results are normalized with respect to
Chase-Lev with a single thread, as in~\cite{MlVS09, MA14}.

\paragraph{Initial length of arrays}
We did both experiments with distinct initial length of arrays
in the array-based algorithms, and distinct lengths of node-array in the
linked-list versions of our algorithms.
The lengths we tried were \(256\), \(4096\),
\(32768\), \(100000\), \(250000\), \(500000\) and \(1,000,000\).
The performance of the algorithms decreased with large initial lengths
but the relative speedup among the algorithms did not change.
Below we discuss the results of the experiments with arrays of
length 256.

\subsection{Results}\label{subsec-results}

We present a summary of the results before explaining them in detail:

\paragraph{Zero cost experiments}
\begin{enumerate}
  \item \NCWSM\@ shows a slightly better performance than Idempotent
    FIFO in the \Put-\Take experiment.  This is expected as both
    algorithms are similar in their \Put/\Take operations, with Idempotent
    FIFO using non-costly fences.  Both algorithms are the ones with the
    best performance in this experiment.  However, \NCWSM\@ exhibited
    better performance in \Put-\Steal experiment. This happens due to the
    fact that the \Steal operation of any other algorithm uses costly
    primitives, like \CAS or \SWAP.

  \item \BNCWSM\@\@ has the worst performance among all algorithms
    in the \Put-\Steal experiment, due to the management of the additional
    array for marking a task as taken. However, this result does not
    precludes the algorithm for exhibiting a competitive performance
    in the next benchmark.
\end{enumerate}

\paragraph{Irregular graph application}

\begin{enumerate}
  \item In general \NCWSM\@ has better performance than any other
    algorithms and in particular it performed better than Idempotent
    FIFO\footnote{We stress the comparison with respect to Idempotent
      FIFO, as both algorithms insert/extract tasks in the same order.} in
    virtually all cases.

  \item \BNCWSM\@ has a lower performance than \NCWSM\@ but still
    competitive respect to Idempotent algorithms.
\end{enumerate}

Below we discuss in detail the results of the zero cost experiments
and the irregular graph application, in both cases omitting the
results of THE Cilk because its performance was similar to that of
Chase-Lev.  Similarly, we omit the results of Idempotent DEQUE since
in general it had the worst performance among the Idempotent
algorithms.

\paragraph{Zero cost experiments}

Figure~\ref{fig-puts-takes} depicts the result of the \Put-\Take
experiment in the Intel Core i7 processor.  The results of THE Cilk
and the idempotent algorithms are similar than those
presented~\cite{MlVS09}.  As for \NCWSM, the time
required for \Put operations was similar than that of Idempotent LIFO,
and slightly faster than Idempotent FIFO. Considering the whole
experiment, \NCWSM was faster than any other algorithm.

The results show a significant speed-up, where the gain is about
21.9\% respect to Chase-Lev, 12.5\% respect Idempotent FIFO and 6\%
respect to Idempotent LIFO.\@

For the case of \BNCWSM, it required about twice the time of any other
algorithm.  This poor performance is due to the use of an extra
boolean array for bounding multiplicity.  Particularly, the \Put
operation writes three entries of an array (two entries of $Tasks$ and
one of the extra boolean array), differently from \NCWSM's \Put
operations that writes only two (both of $Taks$).

\begin{figure}[ht]
  \centering
  \subfloat[Subfigure 1][Results of experiment for puts and takes.] {
    \includegraphics[scale=0.3]{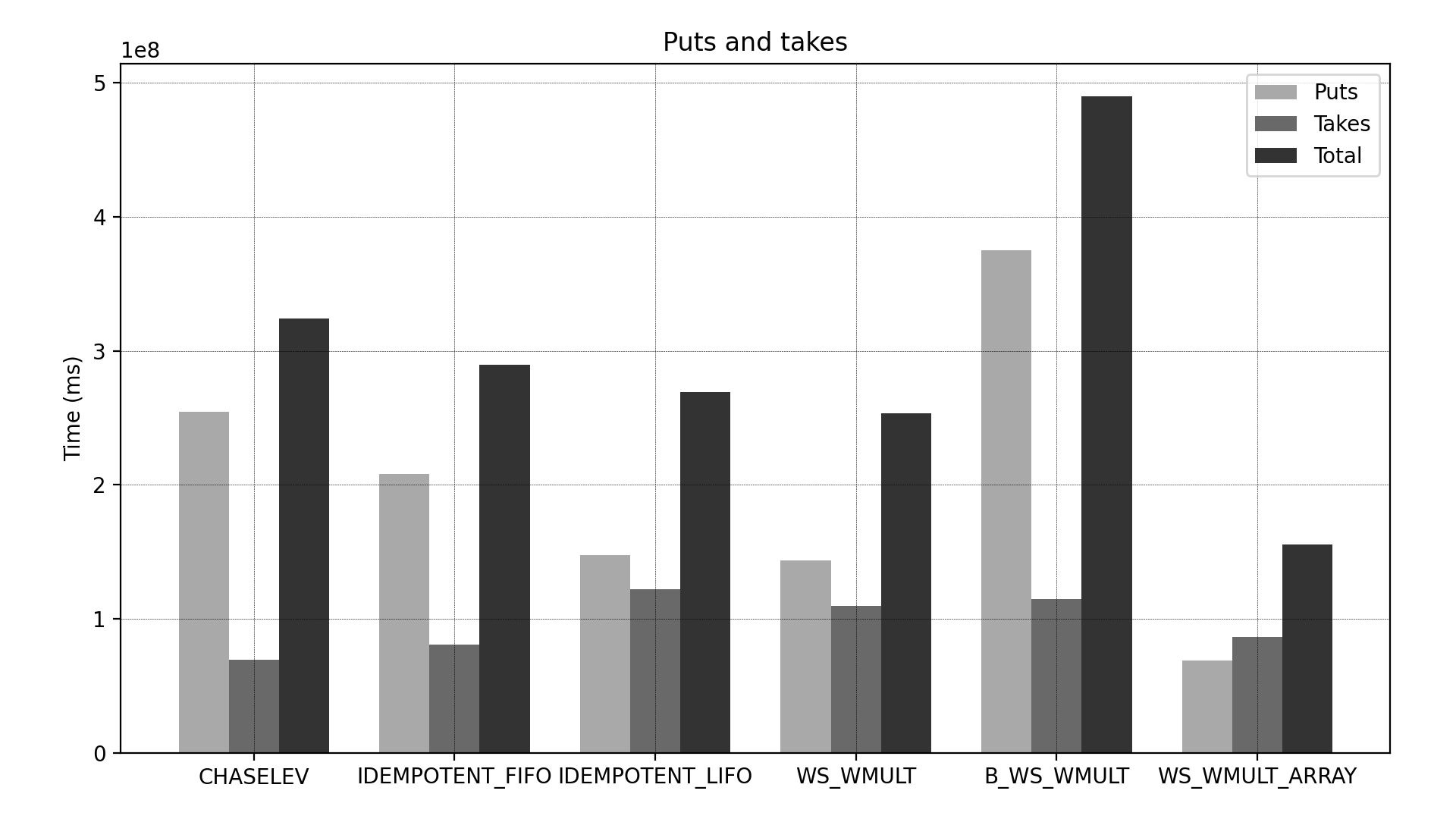}
    \label{fig-puts-takes}
  }
  \subfloat[subfigure 2][Results of experiment for puts and steals.] {
    \includegraphics[scale=0.3]{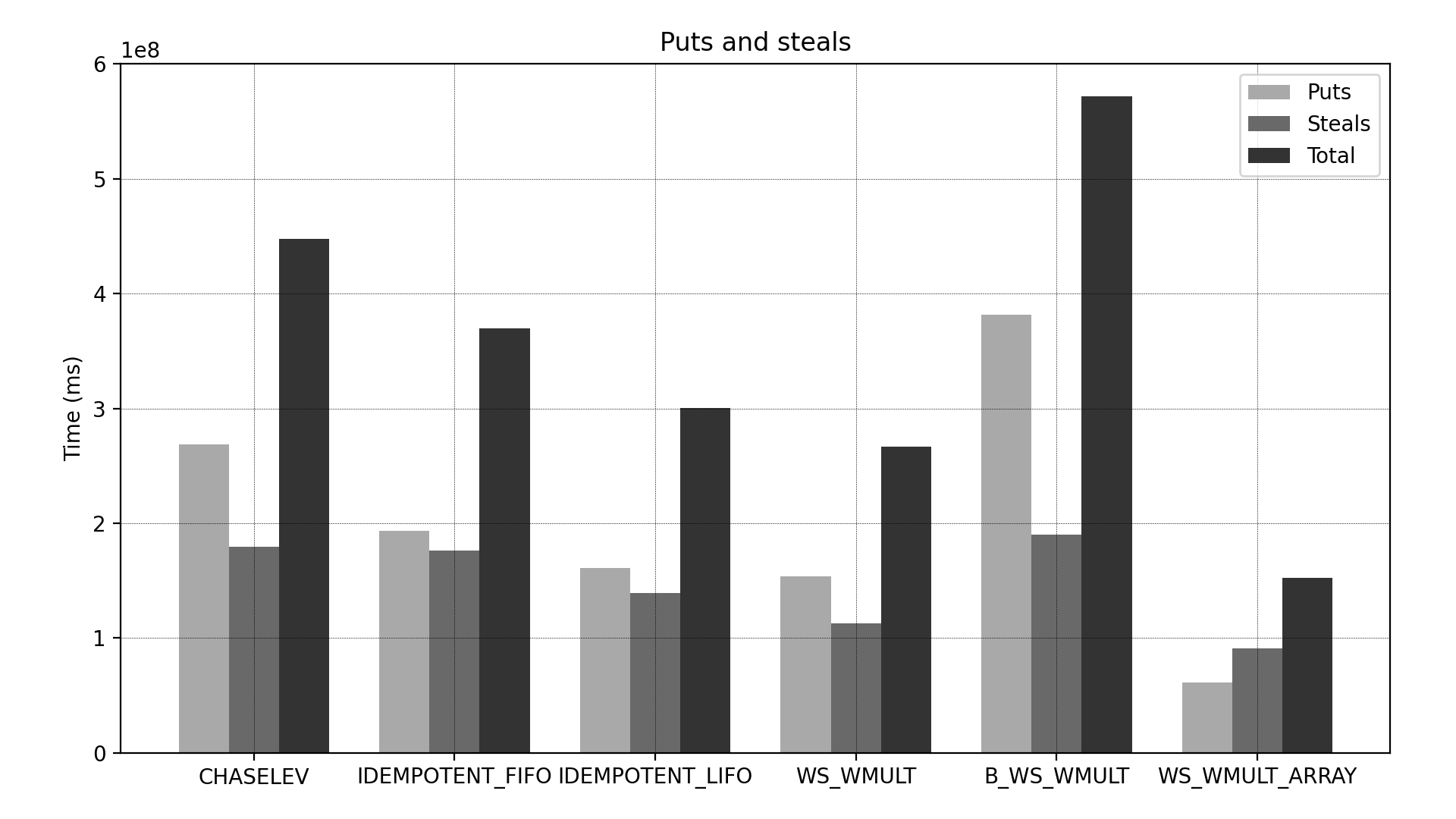}
    \label{fig-puts-steals}
  }
  \caption{Results of experiments for puts, takes and steals.}
\end{figure}

The results of the \Put-\Steal experiment in the processor Intel Core
i7 are shown in Figure~\ref{fig-puts-steals}, where, as expected, \Put
operations exhibited a similar performance as in the \Put-\Take
experiment.

For \Steal operations, we observe that Chase-Lev, Idempotent FIFO, and \NCWSM are
faster than Idempotent LIFO and \BNCWSM. In particular, \NCWSM
is the fastest among all algorithms, which is expected as it does not
use fences and \RMW instructions.  The speed-up is of 37\%, 27.8\% and
11.2\% compared to Chase-Lev, Idempotent FIFO and Idempotent LIFO,
respectively.  As for total time, we observe an speed-up of 40.4\%,
36\%, and 19\% regarding Chase-Lev, Idempotent FIFO and Idempotent
LIFO in each instance.\@ For \BNCWSM\@, we have a similar result than
the one in the \Puts-\Takes\@ experiment.

Finally, the array-based implementation of \NCWSM,
denoted {\sf WS\_WMULT\_ARRAY} the figures, performed much
better than any other algorithm, with a speed-up of up to 65.9\%,
58.7\% and 49.2\% compared to Chase-Lev, Idempotent FIFO and
Idempotent LIFO, for the \Put-\Steal experiment. Similarly, for the
\Put-\Take experiment, the speed-up reached by the array-based
implementation of \NCWSM was of 52.1\% respect to Chase-Lev, 46.3\%
compared to Idempotent FIFO and 42.3\% for Idempotent LIFO.  This
implementation performs better than the linked-list version of \NCWSM
because the latter creates a large number of arrays during the
execution.

In the next experiment, we do not discuss in detail
the result of the array-based implementations of our algorithms
as the linked-list implementations have a similar performance,
and sometimes outperforming them.

\begin{figure}[ht]
  \begin{center}
    \subfloat[Subfigure 3][Speed up for directed random graph.] {
      \includegraphics[scale=0.3]{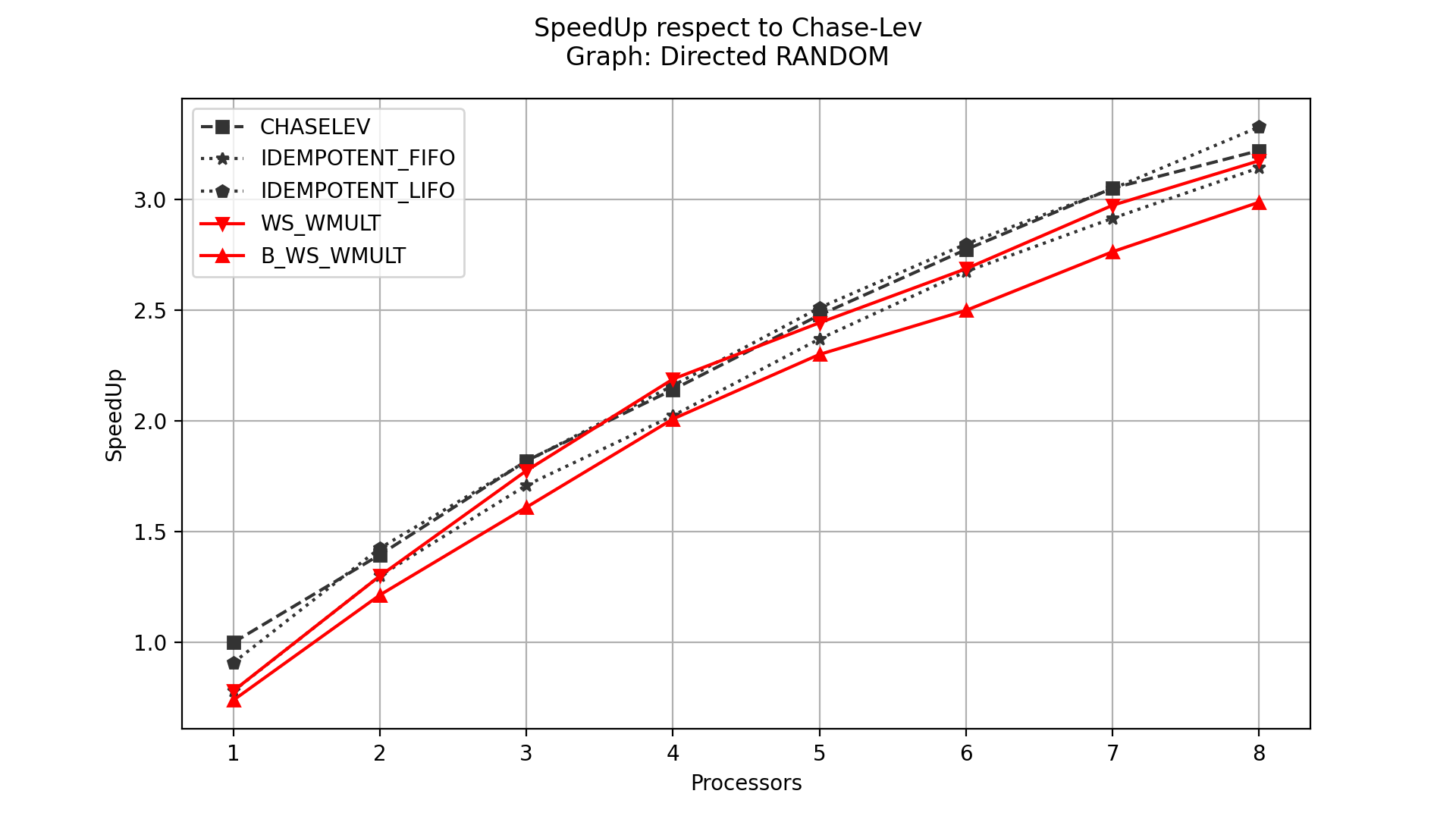}
      \label{directed-random}
    }
    \subfloat[Subfigure 4][Speed up for undirected random graph.] {
      \includegraphics[scale=0.3]{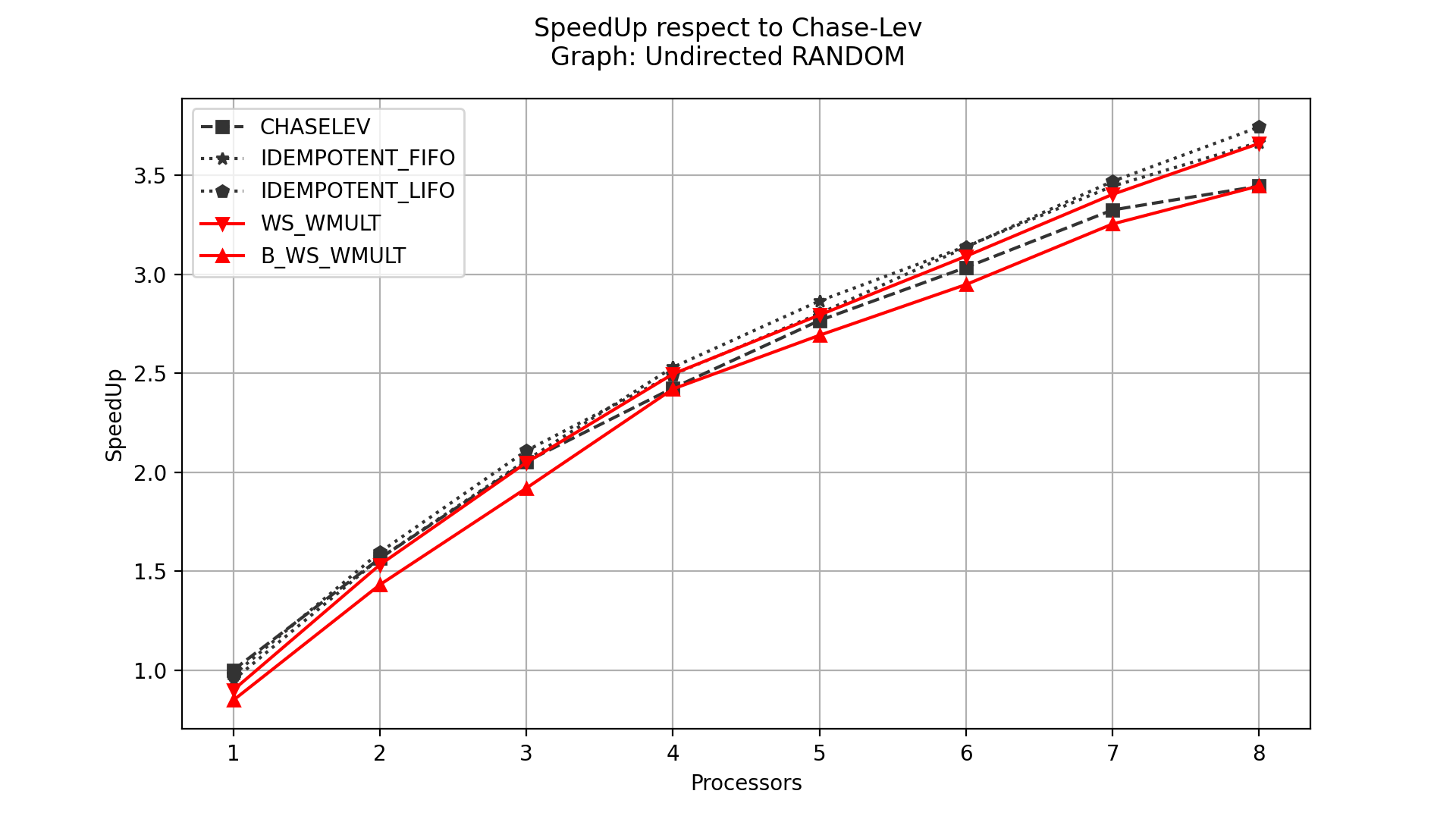}
      \label{undirected-random}
    }
    \qquad
    \subfloat[Subfigure 5][Speed up for directed random graph on Intel
    Xeon.]{
      \includegraphics[scale=0.3]{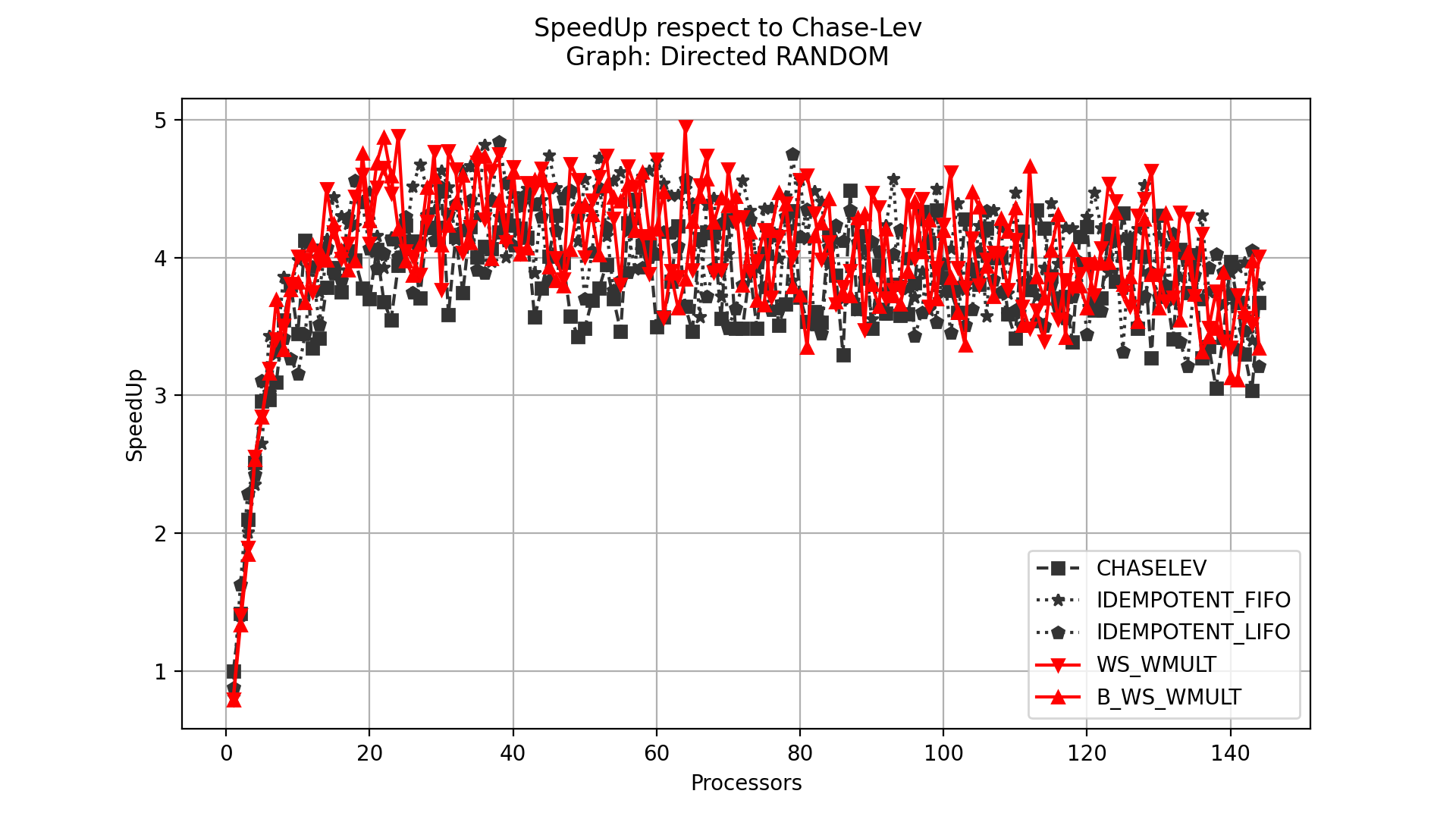}
      \label{cluster-directed-random}
    }
    \subfloat[Subfigure 6][Speed up for undirected random graph on Intel
    Xeon.] {
      \includegraphics[scale=0.3]{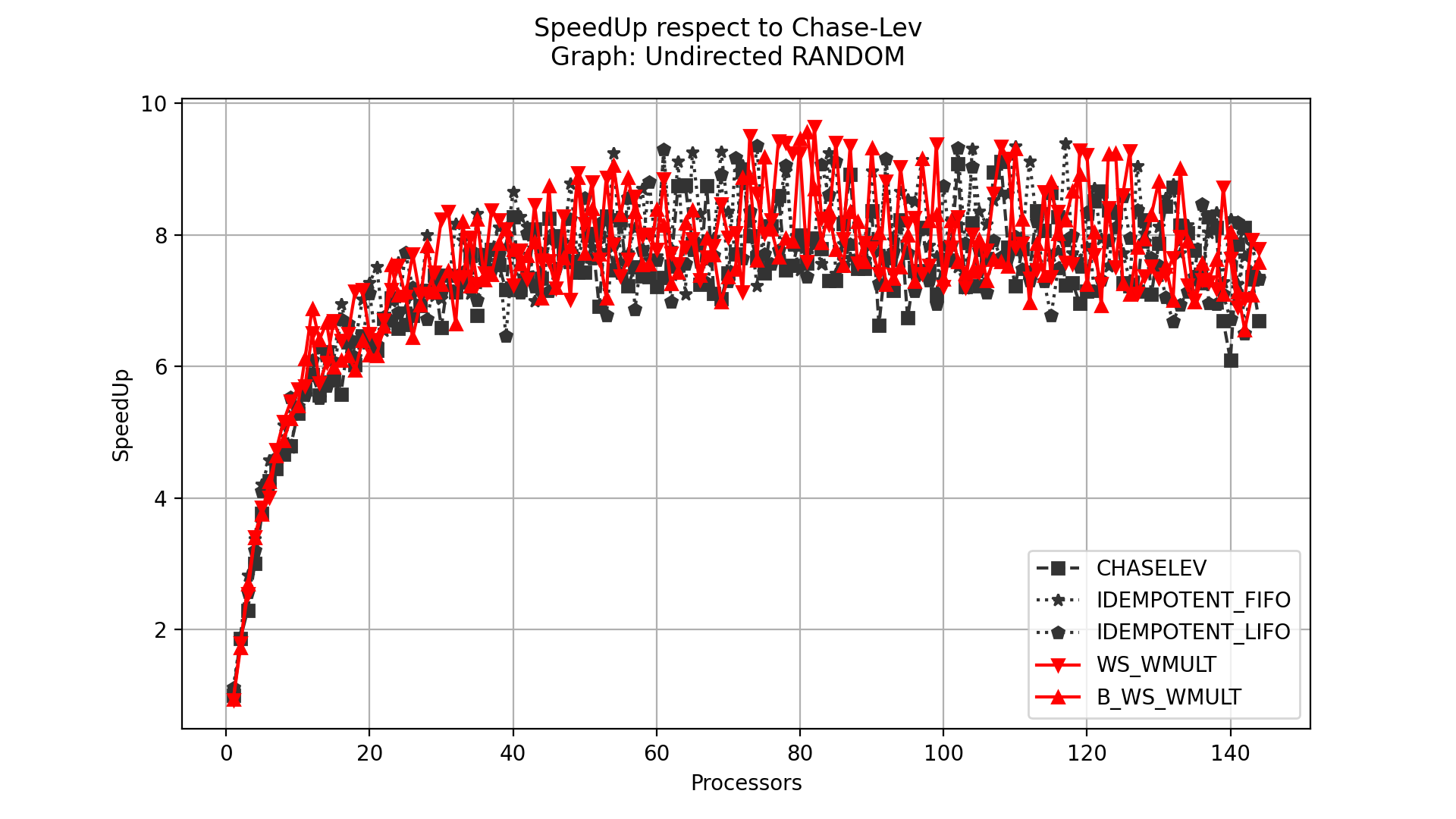}
      \label{cluster-undirected-random}
    }
  \end{center}
  \caption{Speedups for random graphs}
\end{figure}

\paragraph{Irregular graph application}
  Table~\ref{table:range-table} displays the minimum and maximum
  speed-up achieved by each algorithm, normalized with respect to the
  one-thread execution of Chase-Lev, for every graph and processor.
  Overall, \NCWSM\@ had a better performance with maximum speed-up
  of 9.64, corresponding to a gain of 5.39\%, 2.59\%, 2.9\% and 0.72\%
  respect to the maximum speed-up for Chase-Lev, Idempotent FIFO,
  Idempotent LIFO and \BNCWSM), for the undirected random graph in the
  Xeon processor. This speed-up is slightly better than the maximum
  speed-up  achieved by Idempotent LIFO, the best for this graph among
  the other algorithms. The average of the maximum speed-up of
  \NCWSM\@ is 3.92 (corresponding to a gain of 20.15\% respect to
  Chase-Lev,  where the average of the maximum speed-up was of 3.13.
  For Idempotent FIFO, Idempotent LIFO and \BNCWSM the gain was of
  2.04\%, 8.92\% and 3.31\% with an average for the maximum of the
  speed-ups of 3.84, 3.57 and 3.79 respectively. \BNCWSM\@ always
  showed a similar performance than Idempotent FIFO algorithm; we
  expected this behavior as both use expensive \RMW instructions in
  \Steal.  Below we discuss the results for each type of graph.

\paragraph{Random graphs}
  For the directed case (Figure~\ref{directed-random}), we observe
  that \NCWSM and Idempotent LIFO present a slightly better performance
  than the other algorithm. Here, the gains of \NCWSM was of 1.26\% and
  5.97\% for Idempotent FIFO and \BNCWSM.  For the undirected case,
  Idempotent FIFO and LIFO perform sightly better than the other
  algorithms, as shown in figure~\ref{undirected-random}. Using the
  results for the Core i7 processor, we observe the gain for Idempotent
  FIFO and LIFO respect to maximum speed-up of \NCWSM is of 0.27\% and
  2.46\%, respectively, while the gain of \NCWSM respect to Chase-Lev is
  of 5.74\%.  In both cases, \BNCWSM presented the lowest speed-up.  All
  algorithms showed a similar performance for both types of graph in the
  Intel Xeon processor, as shown in Figures~\ref{cluster-directed-random}
  and~\ref{cluster-undirected-random}.

\begin{table*}[ht]
  \centering
  \begin{tabular}{|c|c|c|c|c|c|c|c|}
    \hline
\textbf{Graph} & \textbf{Graph type} & \textbf{Processor} & \textbf{Chase-Lev} & \textbf{FIFO} & \textbf{LIFO} & \textbf{WS-WMULT} & \textbf{B-WS-WMULT} \\
    \hline
\multirow{4}{*}{ \textit{Random} } & \multirow{2}{*}{ Directed } & Core i7 & \(1.0x \sim 3.22x\) & \(0.78x \sim 3.14x\) & \(0.91x \sim 3.33x\) & \(0.78x \sim 3.18x\) & \(0.74x \sim 2.99x\) \\
\cline{3-8}
& & Xeon & \(1.0x \sim 4.49x\) & \(0.81x \sim 4.82x\) & \(0.88x \sim 4.84x\) & \(0.8x \sim 4.95x\) & \(0.79x \sim 4.87x\) \\
\cline{2-8}
& \multirow{2}{*}{ Undirected } & Core i7 & \(1.0x \sim 3.45x\) & \(0.95x \sim 3.67x\) & \(0.98x \sim 3.75x\) & \(0.9x \sim 3.66x\) & \(0.85x \sim 3.45x\) \\
\cline{3-8}
& & Xeon & \(1.0x \sim 9.12x\) & \(0.97x \sim 9.39x\) & \(1.11x \sim 9.36x\) & \(0.92x \sim 9.64x\) & \(0.94x \sim 9.57x\) \\
\hline
\multirow{4}{*}{ \textit{2D Torus} } & \multirow{2}{*}{ Directed } & Core i7 & \(1.0x \sim 2.11x\) & \(0.94x \sim 3.49x\) & \(0.9x \sim 2.73x\) & \(0.92x \sim 3.52x\) & \(0.85x \sim 3.34x\) \\
\cline{3-8}
& & Xeon & \(0.74x \sim 2.63x\) & \(1.21x \sim 5.19x\) & \(0.9x \sim 5.42x\) & \(1.32x \sim 5.33x\) & \(1.26x \sim 5.23x\) \\
\cline{2-8}
& \multirow{2}{*}{ Undirected } & Core i7 & \(1.0x \sim 1.94x\) & \(0.61x \sim 2.14x\) & \(0.84x \sim 1.96x\) & \(0.62x \sim 2.2x\) & \(0.56x \sim 2.06x\) \\
\cline{3-8}
& & Xeon & \(1.0x \sim 1.98x\) & \(0.43x \sim 2.2x\) & \(0.77x \sim 2.33x\) & \(0.43x \sim 2.25x\) & \(0.35x \sim 2.14x\) \\
\hline
\multirow{4}{*}{ \textit{2D60 Torus} } & \multirow{2}{*}{ Directed } & Core i7 & \(1.0x \sim 2.06x\) & \(0.59x \sim 2.21x\) & \(0.87x \sim 2.15x\) & \(0.6x \sim 2.32x\) & \(0.56x \sim 2.13x\) \\
\cline{3-8}
& & Xeon & \(1.0x \sim 1.95x\) & \(0.66x \sim 2.38x\) & \(0.86x \sim 2.09x\) & \(0.64x \sim 2.39x\) & \(0.6x \sim 2.43x\) \\
\cline{2-8}
& \multirow{2}{*}{ Undirected } & Core i7 & \(1.0x \sim 2.19x\) & \(0.77x \sim 2.71x\) & \(0.88x \sim 2.55x\) & \(0.78x \sim 2.81x\) & \(0.73x \sim 2.51x\) \\
\cline{3-8}
& & Xeon & \(1.0x \sim 2.57x\) & \(0.56x \sim 3.66x\) & \(0.8x \sim 3.29x\) & \(0.57x \sim 3.79x\) & \(0.56x \sim 3.61x\) \\
\hline
\multirow{4}{*}{ \textit{3D Torus} } & \multirow{2}{*}{ Directed } & Core i7 & \(1.0x \sim 2.21x\) & \(0.92x \sim 3.17x\) & \(0.88x \sim 2.79x\) & \(0.91x \sim 3.47x\) & \(0.86x \sim 3.26x\) \\
\cline{3-8}
& & Xeon & \(0.7x \sim 3.52x\) & \(0.92x \sim 5.75x\) & \(1.02x \sim 4.64x\) & \(1.01x \sim 5.75x\) & \(0.94x \sim 5.65x\) \\
\cline{2-8}
& \multirow{2}{*}{ Undirected } & Core i7 & \(1.0x \sim 2.5x\) & \(0.73x \sim 2.51x\) & \(0.9x \sim 2.56x\) & \(0.67x \sim 2.58x\) & \(0.62x \sim 2.52x\) \\
\cline{3-8}
& & Xeon & \(1.0x \sim 3.77x\) & \(0.83x \sim 4.94x\) & \(0.94x \sim 3.94x\) & \(0.85x \sim 4.87x\) & \(0.76x \sim 5.05x\) \\
\hline
\multirow{4}{*}{ \textit{3D40 Torus} } & \multirow{2}{*}{ Directed } & Core i7 & \(1.0x \sim 2.95x\) & \(0.88x \sim 3.19x\) & \(0.91x \sim 3.06x\) & \(0.86x \sim 3.26x\) & \(0.82x \sim 2.99x\) \\
\cline{3-8}
& & Xeon & \(1.0x \sim 3.63x\) & \(0.8x \sim 4.46x\) & \(0.88x \sim 3.99x\) & \(0.8x \sim 4.47x\) & \(0.8x \sim 4.43x\) \\
\cline{2-8}
& \multirow{2}{*}{ Undirected } & Core i7 & \(1.0x \sim 2.65x\) & \(0.8x \sim 3.05x\) & \(0.89x \sim 2.84x\) & \(0.78x \sim 3.02x\) & \(0.72x \sim 2.78x\) \\
\cline{3-8}
& & Xeon & \(1.0x \sim 3.67x\) & \(0.85x \sim 4.78x\) & \(0.96x \sim 3.95x\) & \(0.79x \sim 4.98x\) & \(0.85x \sim 4.81x\) \\
\hline
  \end{tabular}

  \caption{\label{table:range-table} Summary of speedups performed by
    each algorithm. We show the maximum and the minimum obtained in
    each experiment.}
\end{table*}

\paragraph{2D Torus}
  In the case of 2D Torus (and similarly for the 2D60 Torus), we observe a
  pattern: our algorithms showed a bad performance with one or two threads,
  but when the number of threads increases, our algorithms
  improved the speed up of Idempotent LIFO and Chase-Lev. Additionally,
  the performance showed by Idempotent FIFO is between the registered by
  \NCWSM\@ and \BNCWSM. We can see this behavior in
  Figures~\ref{directed-torus2d},~\ref{undirected-torus2d},
  \ref{directed-torus2d60} and~\ref{undirected-torus2d60}.  In the
  results of the experiments in the Core i7 processor, we observe the
  following: for the 2D Directed Torus, the gains of \NCWSM respect to
  the maximum average were 40\%, 0.85\% and 22.4\% for Chase-Lev,
  Idempotent FIFO and LIFO respectively. Making a similar comparison,
  for the 2D Undirected Torus, we observed gains of 11.82\%, 2.73\% and
  10.91\%. For the 2D60 Directed Torus, the gains of \NCWSM were
  11.21\%, 4.74\% and 7.33\%. For the 2D60 Undirected
  Torus, the gains with \NCWSM were 22.06\%, 3.56\% and 9.25\%.

  A similar behavior can be observed in the experiments performed in
  the Intel Xeon cluster. In most of the of executions, \NCWSM\@ has a
  better performance than other algorithms, see the
  Figures~\ref{cluster-directed-torus2d},~\ref{cluster-undirected-torus2d},
  ~\ref{cluster-directed-torus2d60}
  and~\ref{cluster-undirected-torus2d60}. We can also observe a
  noticeable loss of performance of Chase-Lev
  over the directed 2D Torus, sometimes a with
  performance that is worse than if it were ran using a single thread
  (Figure~\ref{cluster-directed-torus2d}). It similarly happens for the
  directed and undirected 2D60 Torus.

  \begin{figure}[ht]
    \begin{center}
      \subfloat[Subfigure 7][Speed up for directed torus 2D graph on
      Intel Core i7.] {
        \includegraphics[scale=0.3]{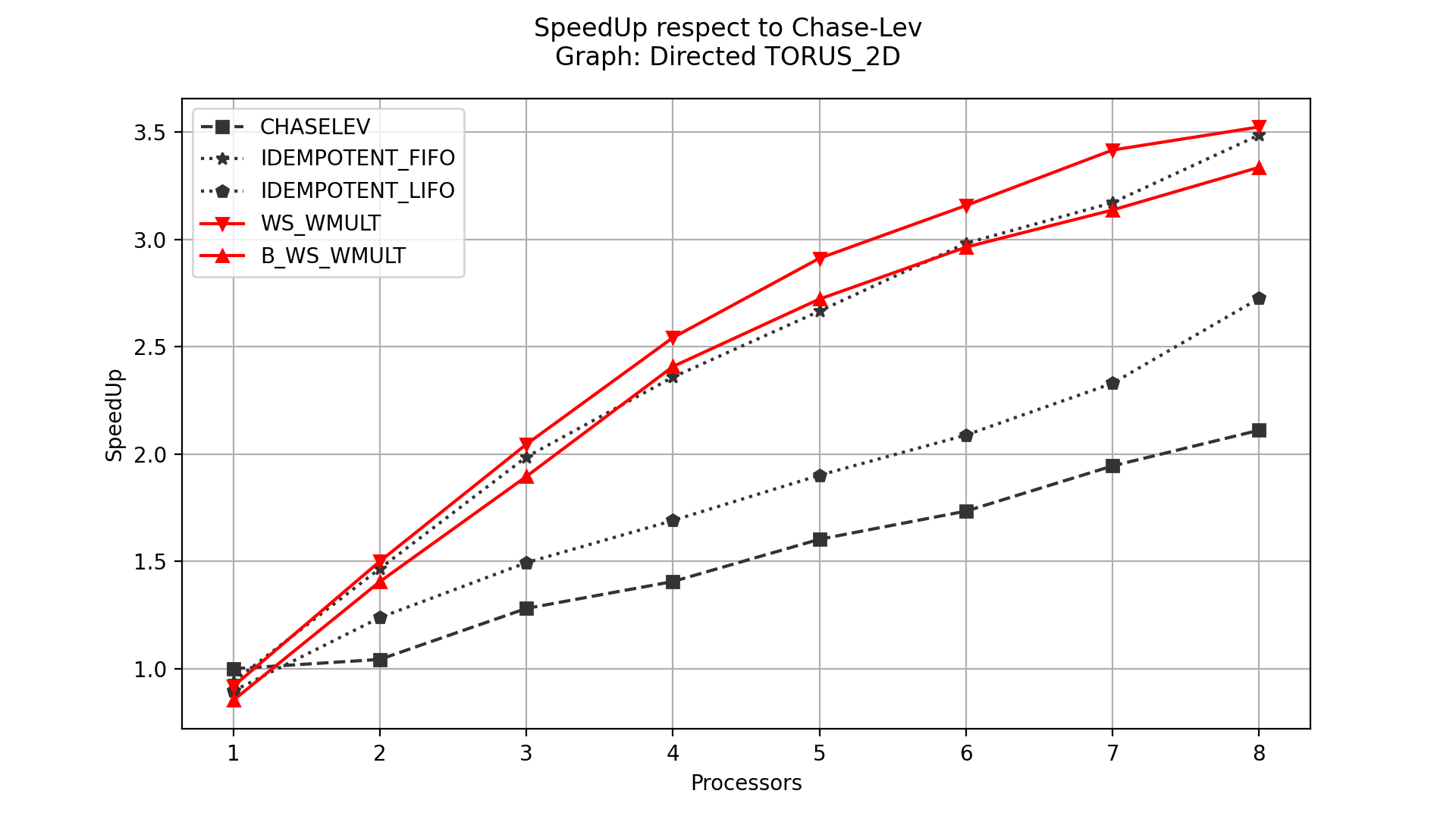}
        \label{directed-torus2d}
      }
      \subfloat[Subfigure 8][Speed up for undirected torus 2D graph on
      Intel Core i7.]{
        \includegraphics[scale=0.3]{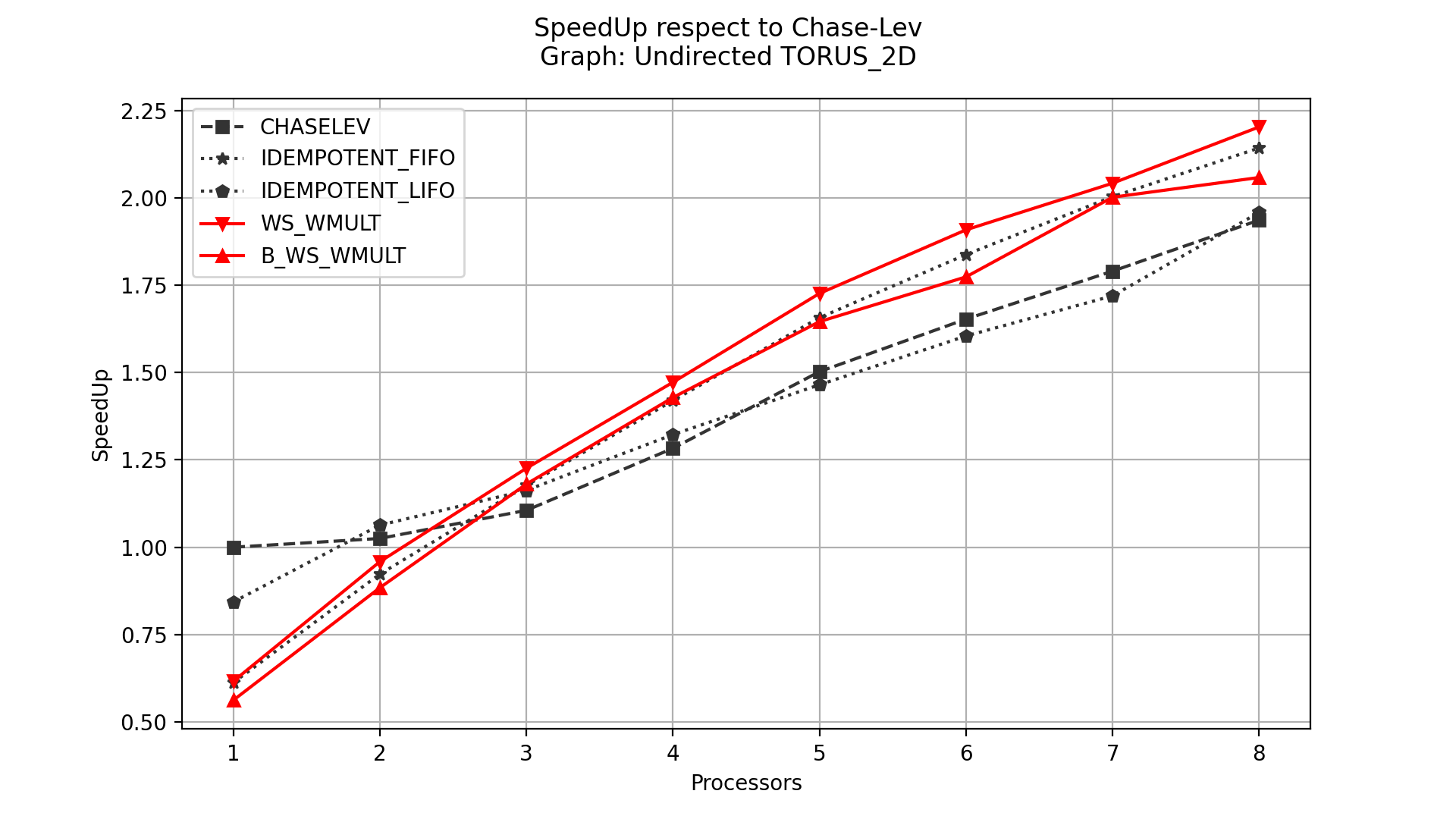}
        \label{undirected-torus2d}
      }
      \qquad
      \subfloat[Subfigure 9][Speed up for directed 2D60 torus graph on
      Intel Core i7.] {
        \includegraphics[scale=0.3]{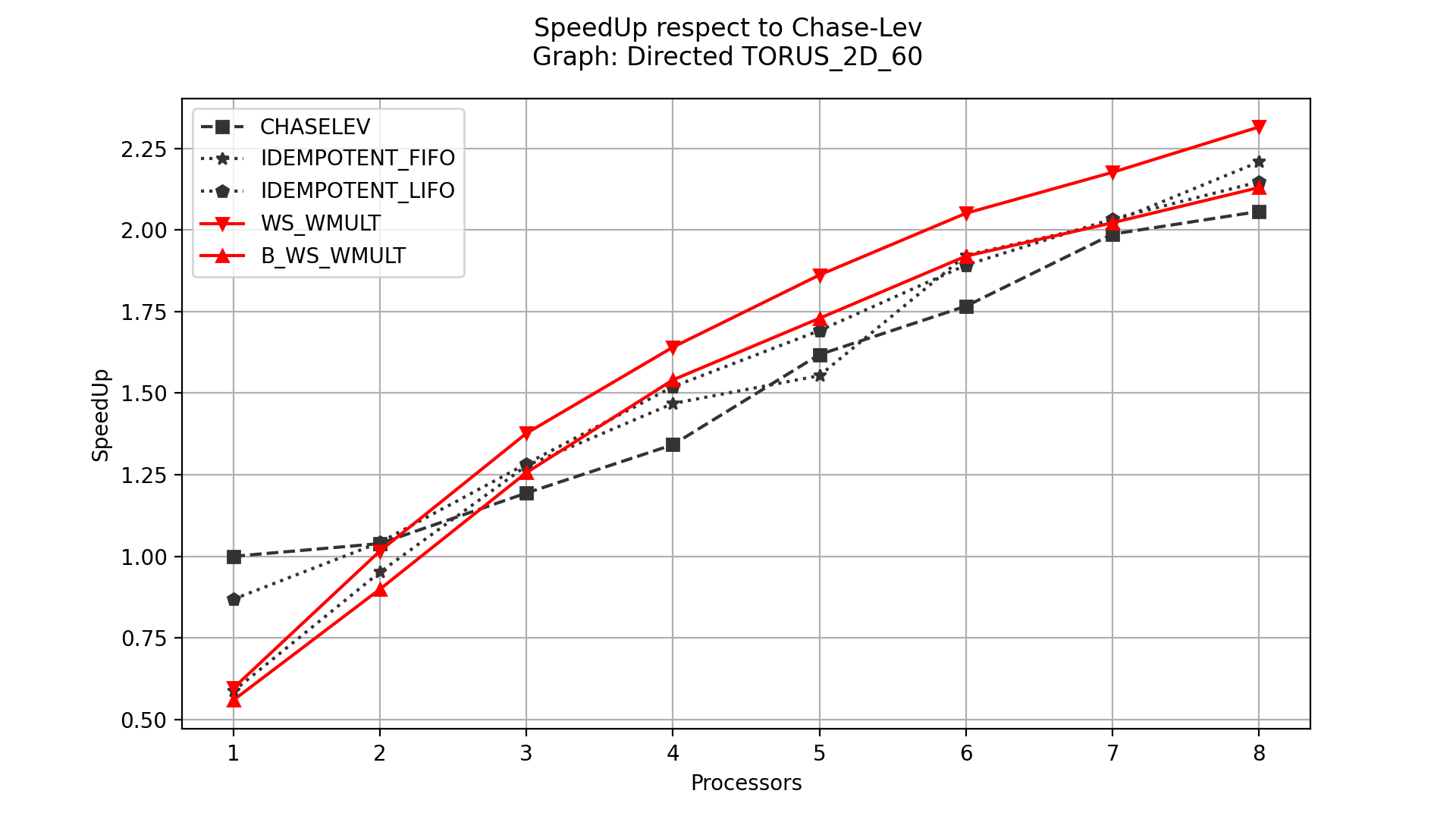}
        \label{directed-torus2d60}
      }
      \subfloat[Subfigure 10][Speed up for undirected 2D60 torus graph
      on Intel Core i7.] {
        \includegraphics[scale=0.3]{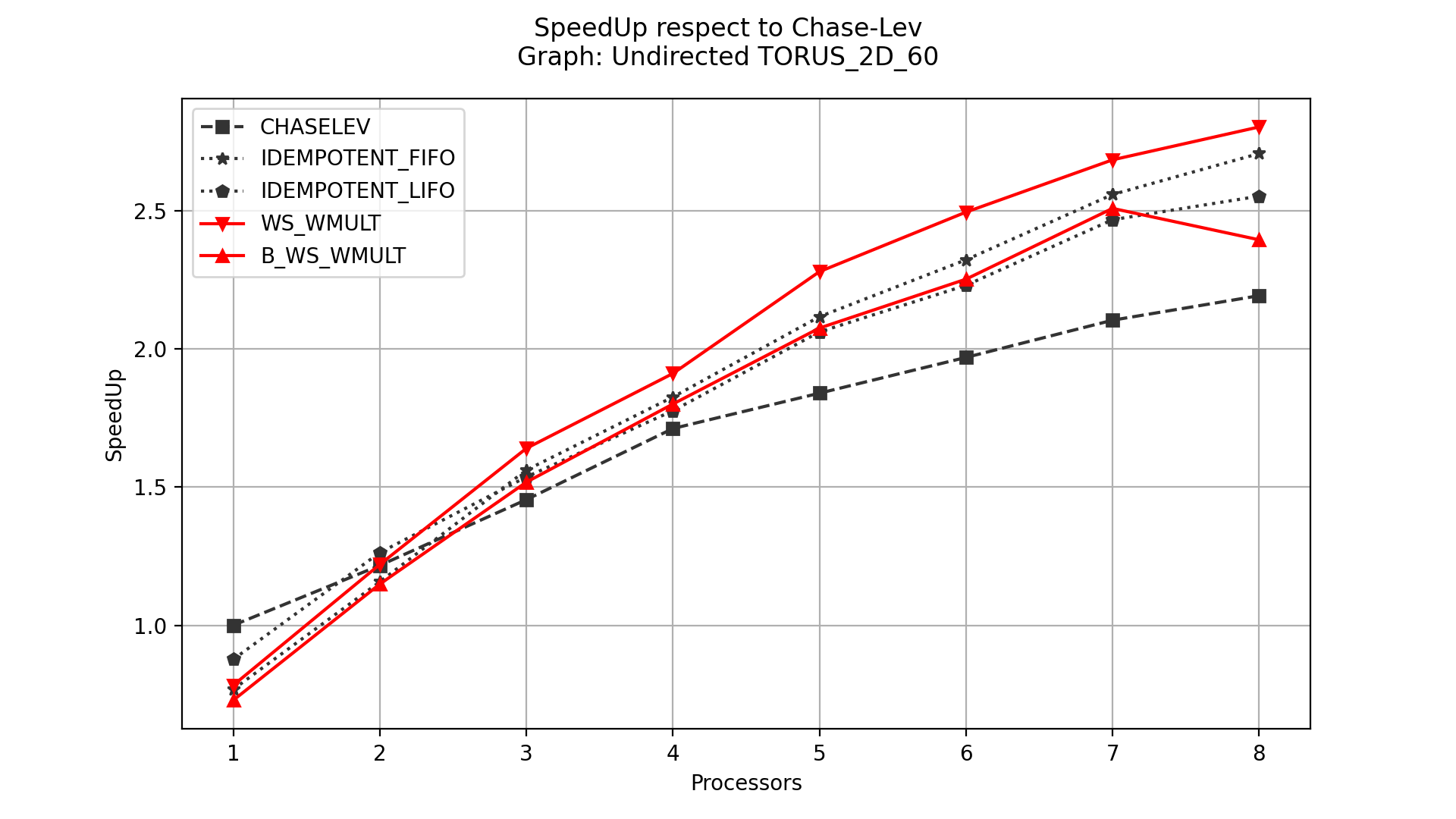}
        \label{undirected-torus2d60}
      }
    \end{center}
    \caption{Speedups for 2D Torus on Intel Core i7}
  \end{figure}

  \begin{figure}[ht]
    \centering
      \subfloat[Subfigure 11][Speed up for directed torus 2D graph on
      Intel Xeon cluster.] {
        \includegraphics[scale=0.3]{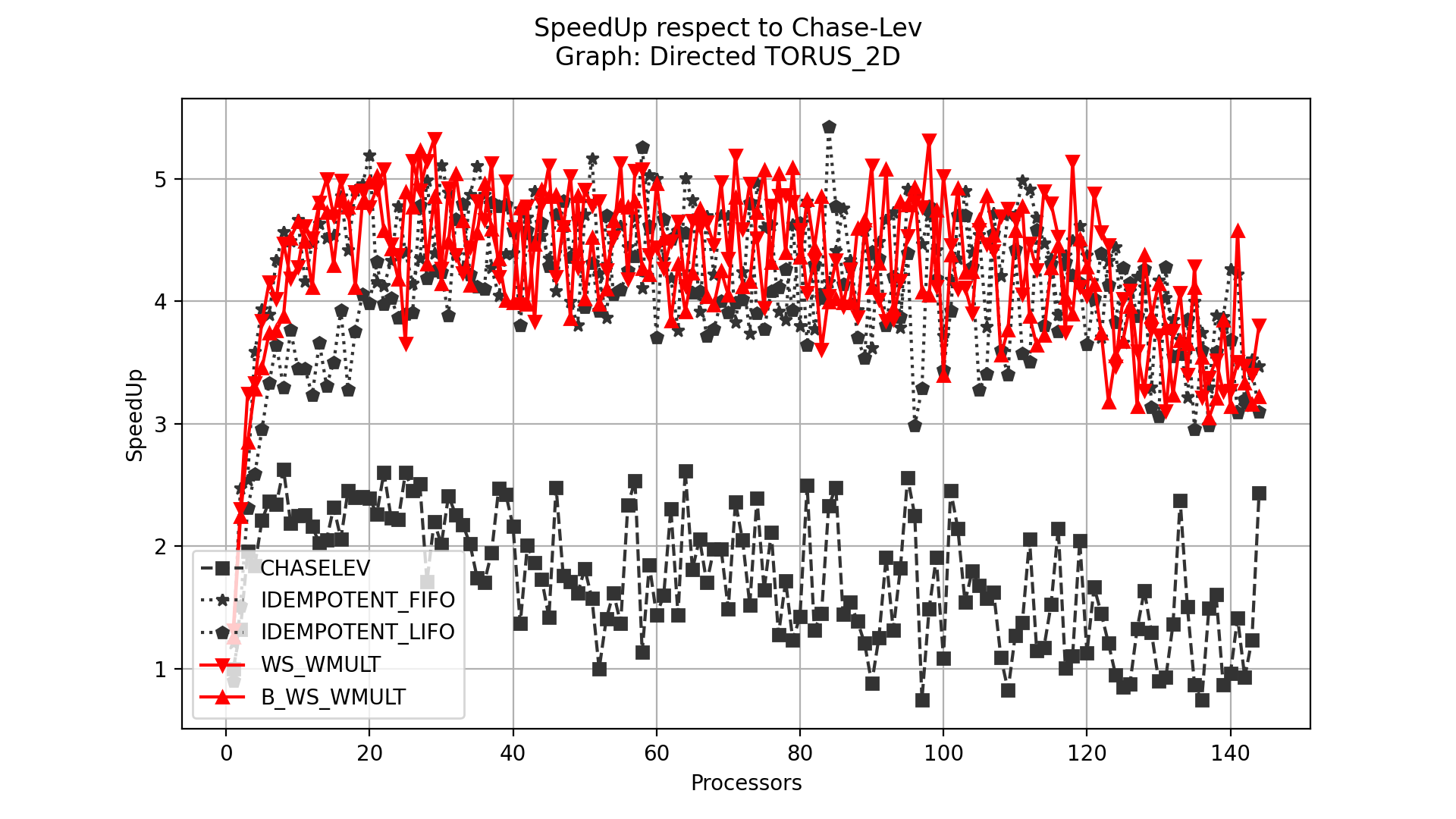}
        \label{cluster-directed-torus2d}
      }
      \subfloat[Subfigure 12][Speed up for undirected torus 2D graph
      on Intel Xeon cluster.] {
        \includegraphics[scale=0.33]{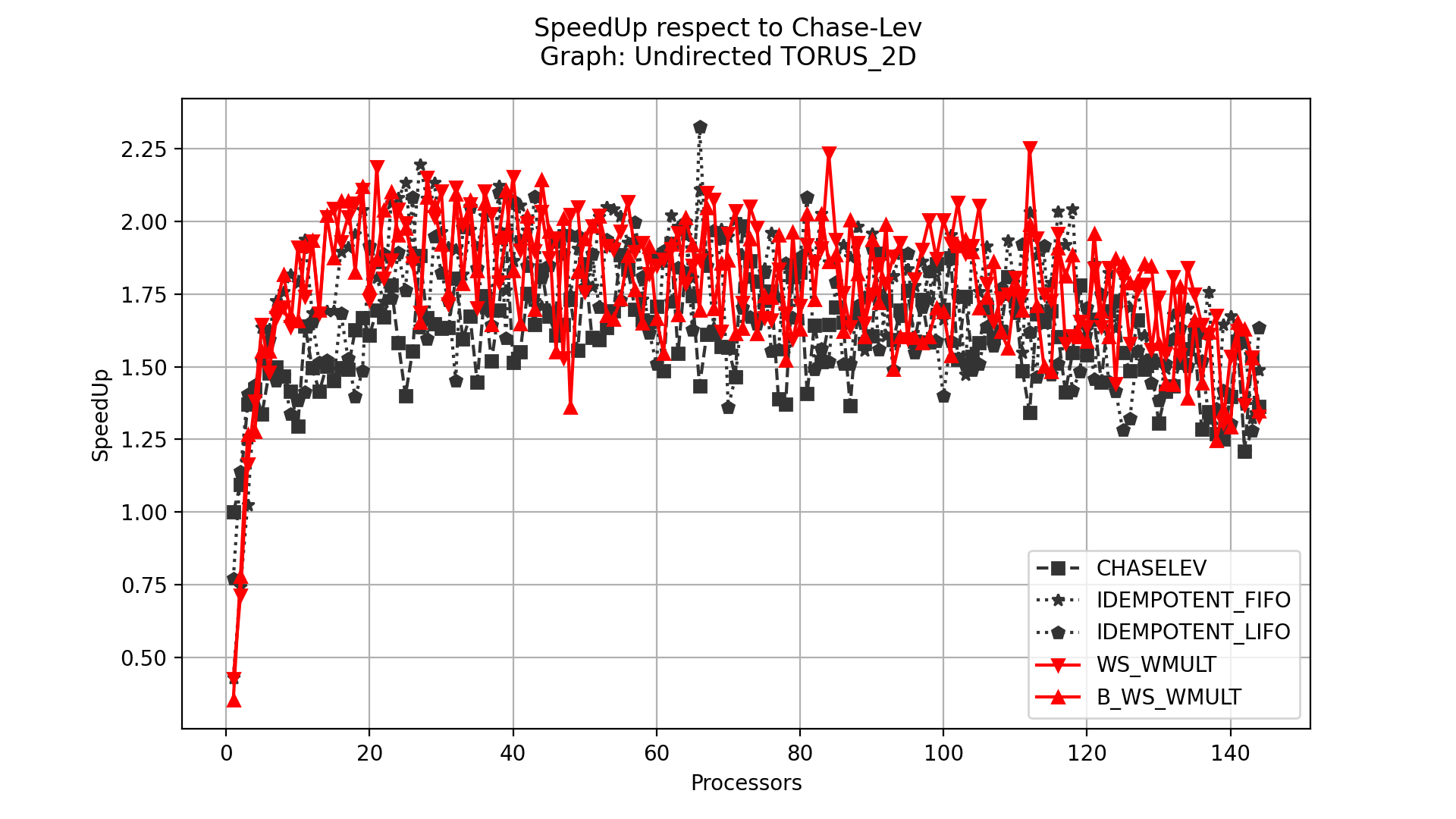}
        \label{cluster-undirected-torus2d}
      }
      \qquad
      \subfloat[Subfigure 13][Speed up for directed 2D60 torus graph
      on Intel Xeon cluster.] {
        \includegraphics[scale=0.3]{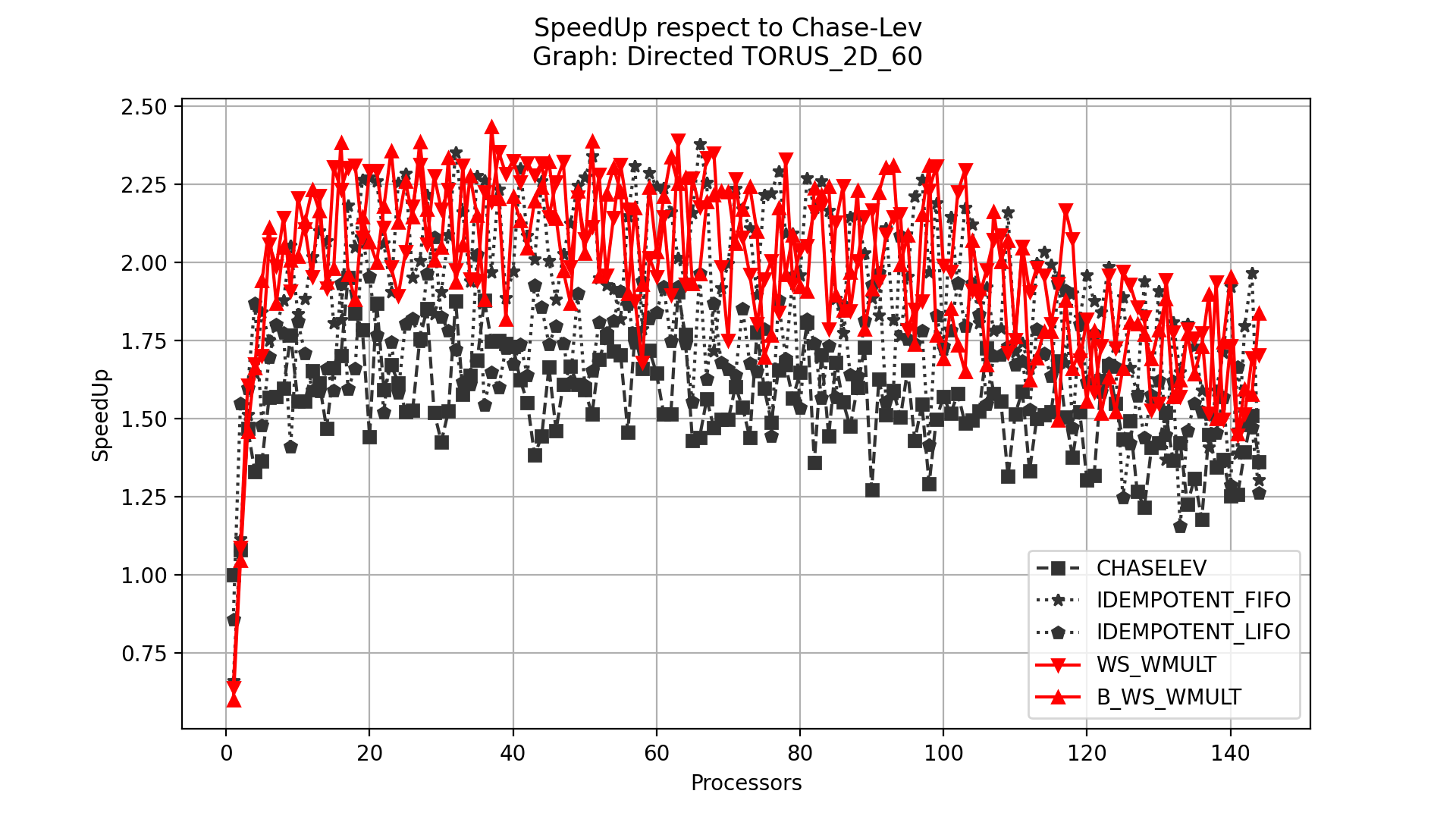}
        \label{cluster-directed-torus2d60}
      }
      \subfloat[Subfigure 14][Speed up for undirected 2D60 torus graph
      on Intel Xeon cluster.]{
        \includegraphics[scale=0.33]{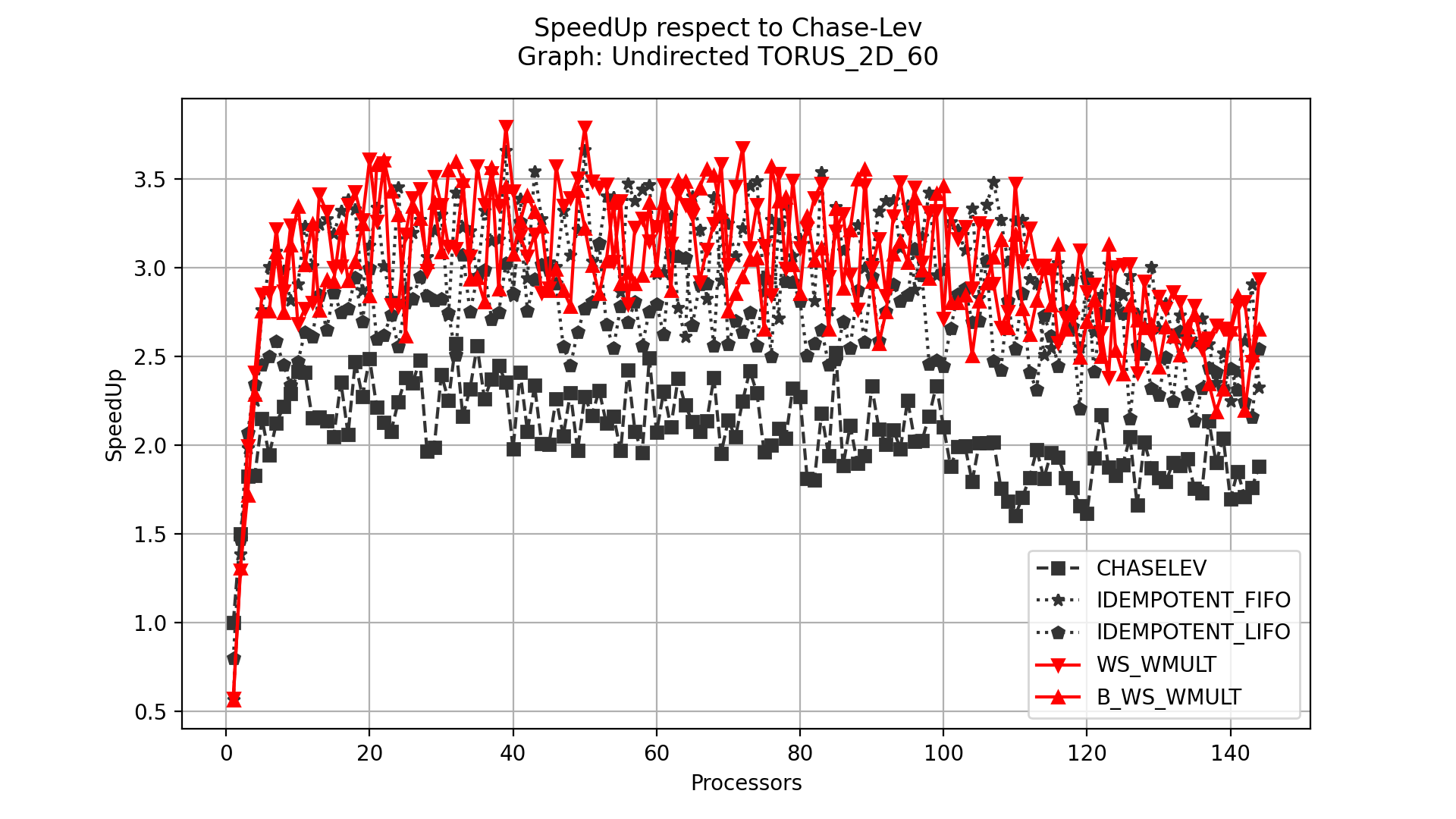}
        \label{cluster-undirected-torus2d60}
      }
    \caption{Speedups for 2D Torus on Intel Xeon}
  \end{figure}

\paragraph{3D Torus}
  For the directed 3D Torus and the 3D40 Torus (directed and
  undirected), we noted a behavior similar to that of the 2D Torus
  (Figures~\ref{directed-torus3d},~\ref{directed-torus3d40},~\ref{undirected-torus3d40}).
  In these experiments, we observe the following gains for the maximum
  speed-up performed by \NCWSM in the directed 3D Torus: 36.61\%, 8.65\%
  and 19.6\% respect to Chase-Lev, Idempotent FIFO and Idempotent LIFO
  respectively.
  For the directed 3D40 Torus, the gains were 9.51\%, 2.15\% and 6.13\%,
  and for the undirected version, 12.25\%, -0.99\%\footnote{In this case, Idempotent
  FIFO has a sightly better maximum speed-up than \NCWSM} and 5.95\%.
  For the undirected 3D Torus, the Idempotent FIFO presents a slightly
  better performance than \NCWSM\@ algorithm, while
  \BNCWSM\@ has a sightly worse performance than any other
  algorithms (Figure~\ref{undirected-torus3d}). Here, the gains of
  \NCWSM for the maximum speed-ups, were of 3.1\%, 2.71\% and 0.78\%
  respect to Chase-Lev, Idempotent FIFO and Idempotent LIFO.

  In the experiments in the Intel Xeon, we observe that our
  algorithms, together with Idempotent FIFO, exhibit a better performance for
  the four versions of 3D torus
  (Figures~\ref{cluster-directed-torus3d},~\ref{cluster-undirected-torus3d},
  ~\ref{cluster-directed-torus3d40}
  and~\ref{cluster-undirected-torus3d40}).

  \begin{figure}[ht]
    \begin{center}
      \subfloat[Subfigure 15][Speed up for directed torus 3D graph on Intel Core
      i7.] {
        \includegraphics[scale=0.3]{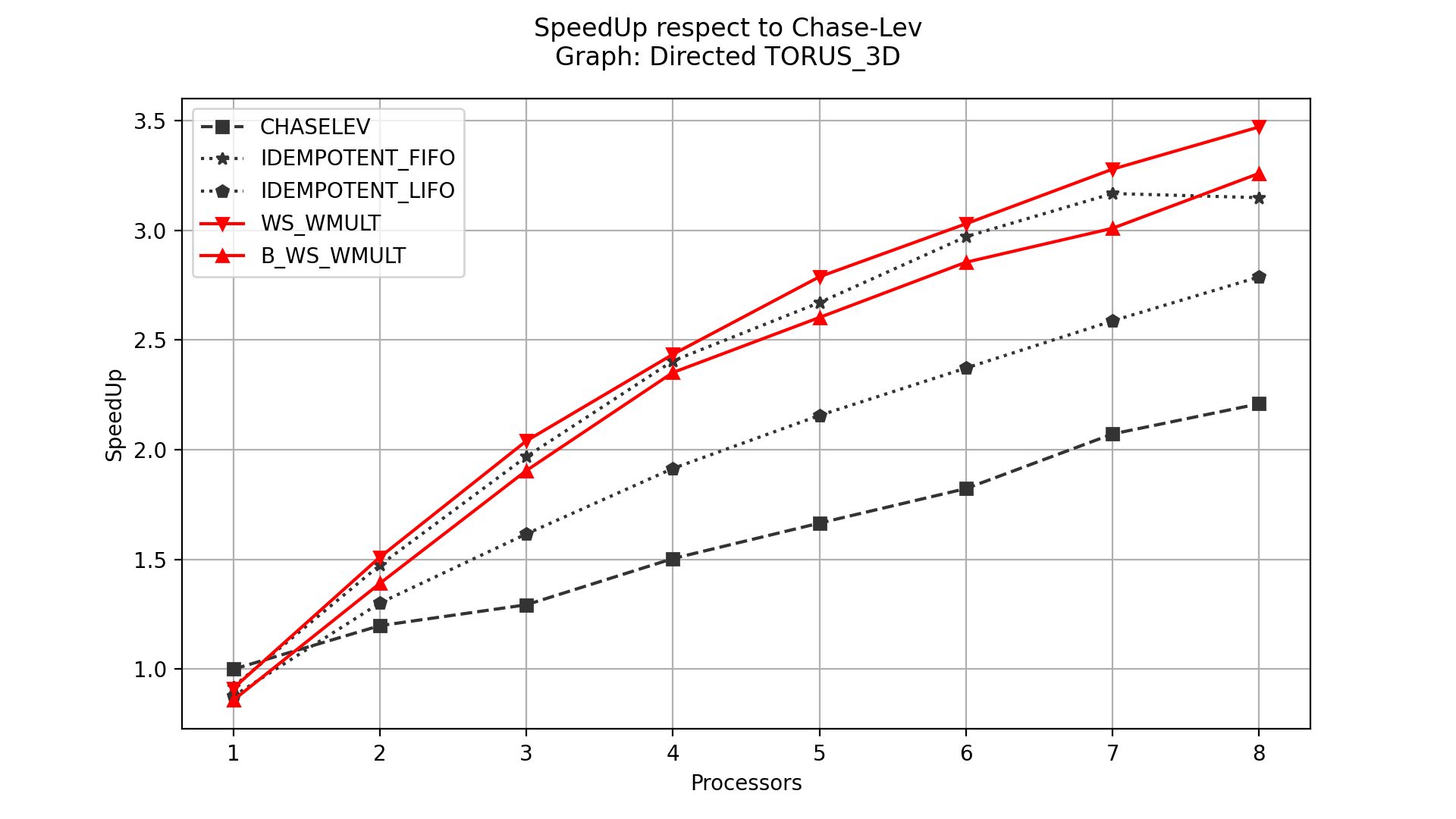}
        \label{directed-torus3d}
      }
      \subfloat[Subfigure 16][Speed up for undirected torus 3D graph on Intel Core
      i7.] {
        \includegraphics[scale=0.3]{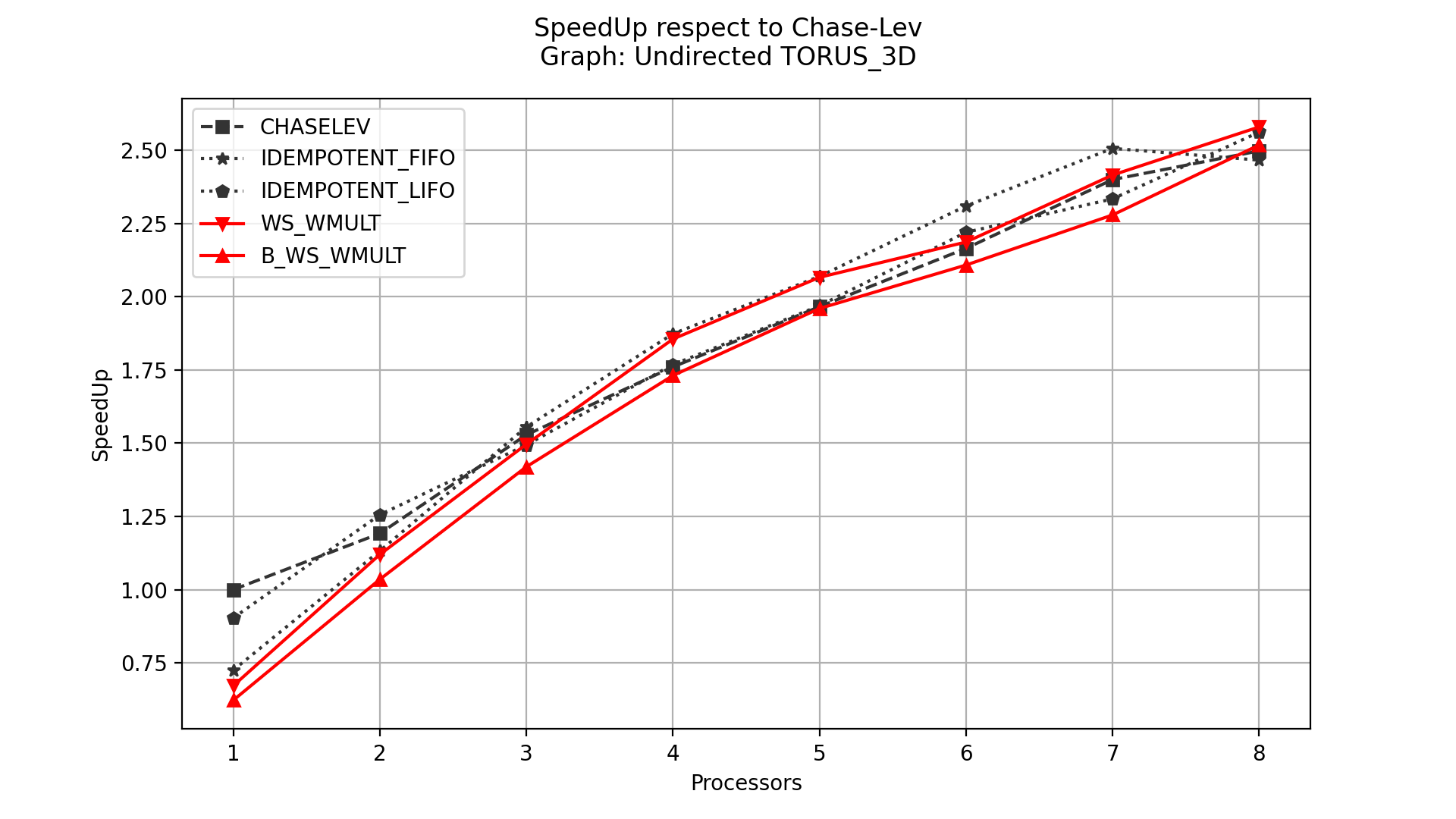}
        \label{undirected-torus3d}
      }
      \qquad
      \subfloat[Subfigure 17][Speed up for directed torus 3D40 graph on Intel Core
      i7.] {
        \includegraphics[scale=0.3]{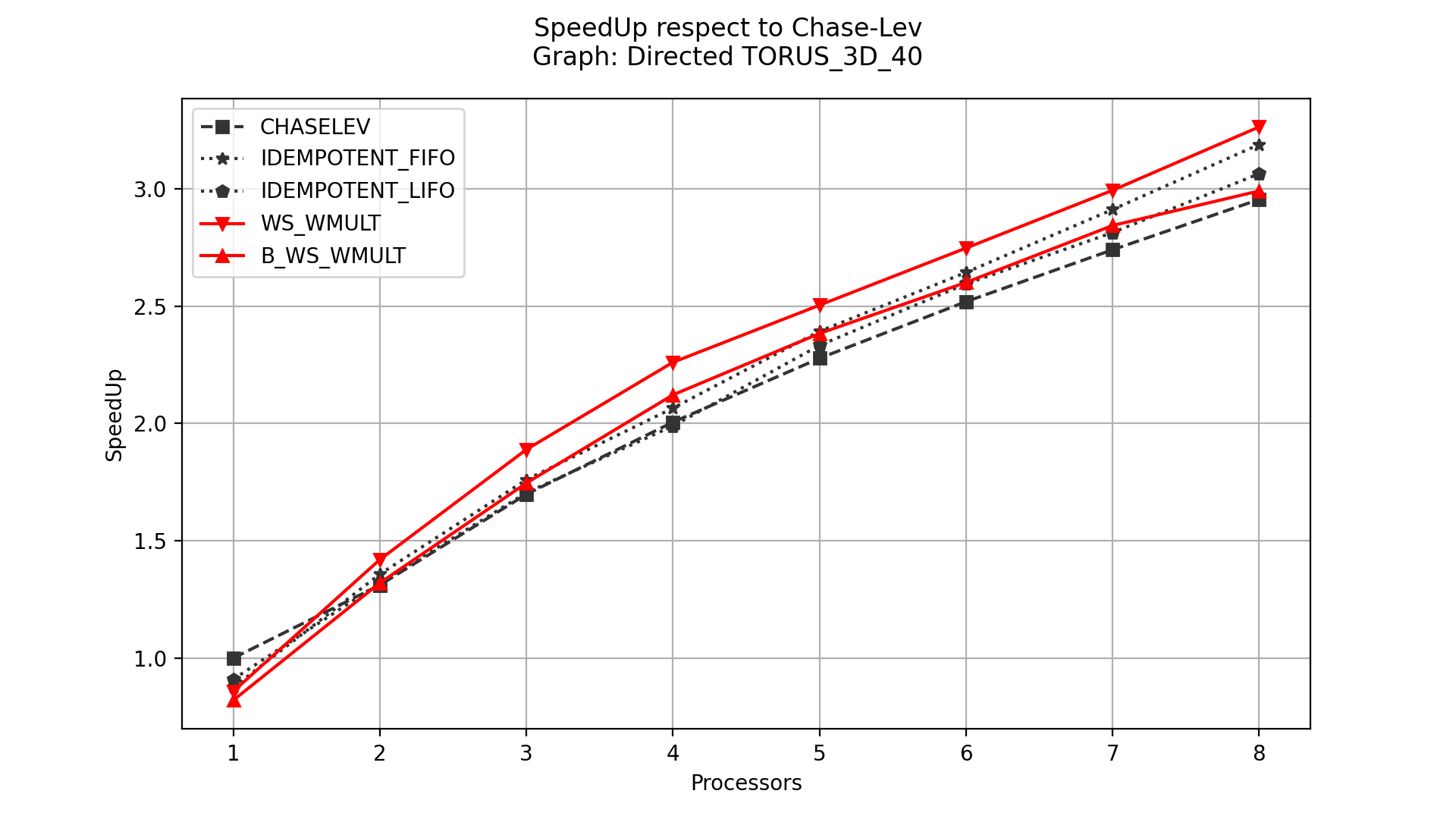}
        \label{directed-torus3d40}
      }
      \subfloat[Subfigure 18][Speed up for undirected torus 3D40 graph on Intel Core
      i7.] {
        \includegraphics[scale=0.3]{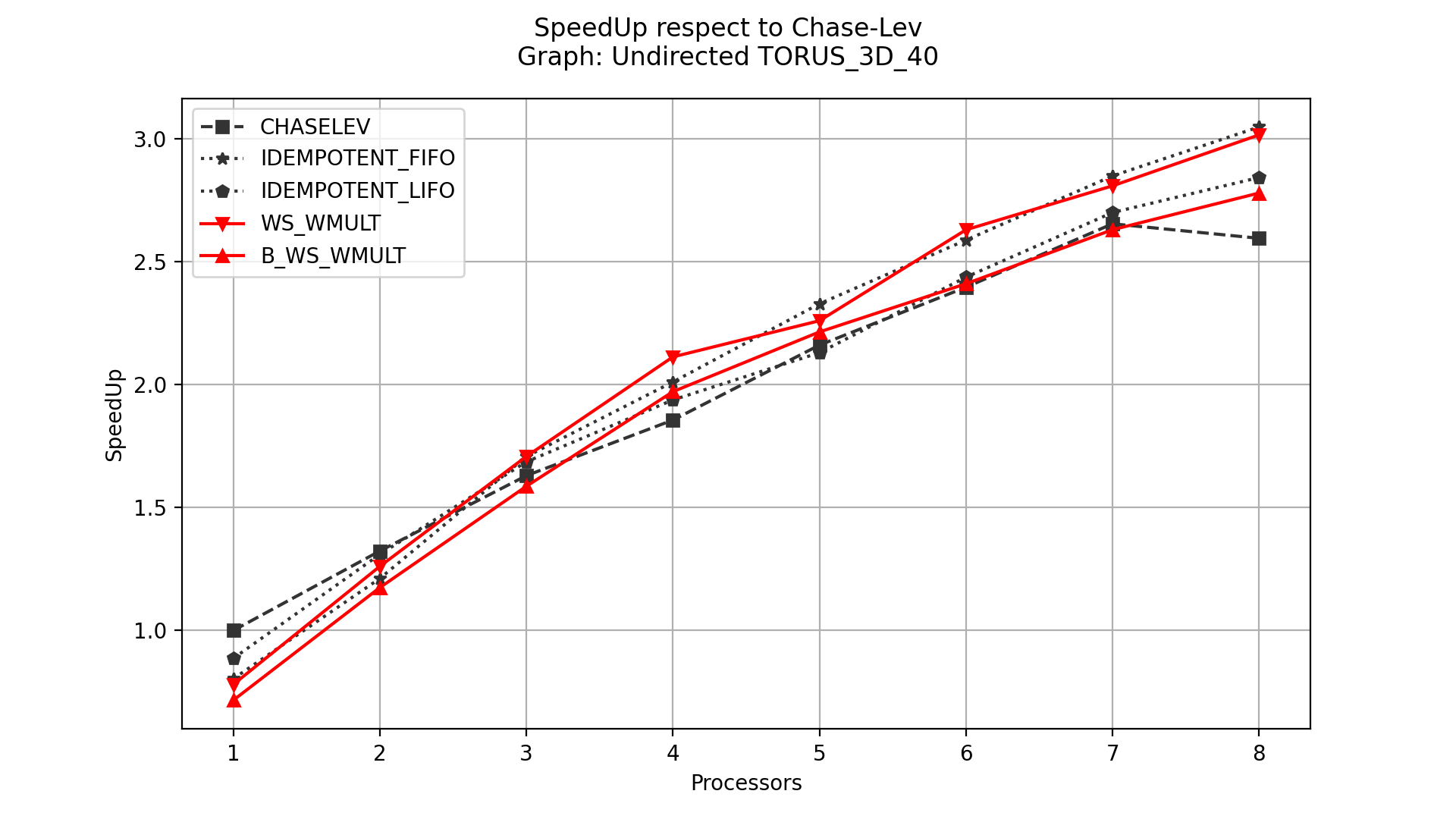}
        \label{undirected-torus3d40}
      }
    \end{center}
        \caption{Speedups for 3D Torus on Intel Core i7}
  \end{figure}

\begin{figure}[ht]
  \begin{center}
    \subfloat[Subfigure 19][Speed up for directed torus 3D graph on Intel
    Xeon.]{
    \includegraphics[scale=0.33]{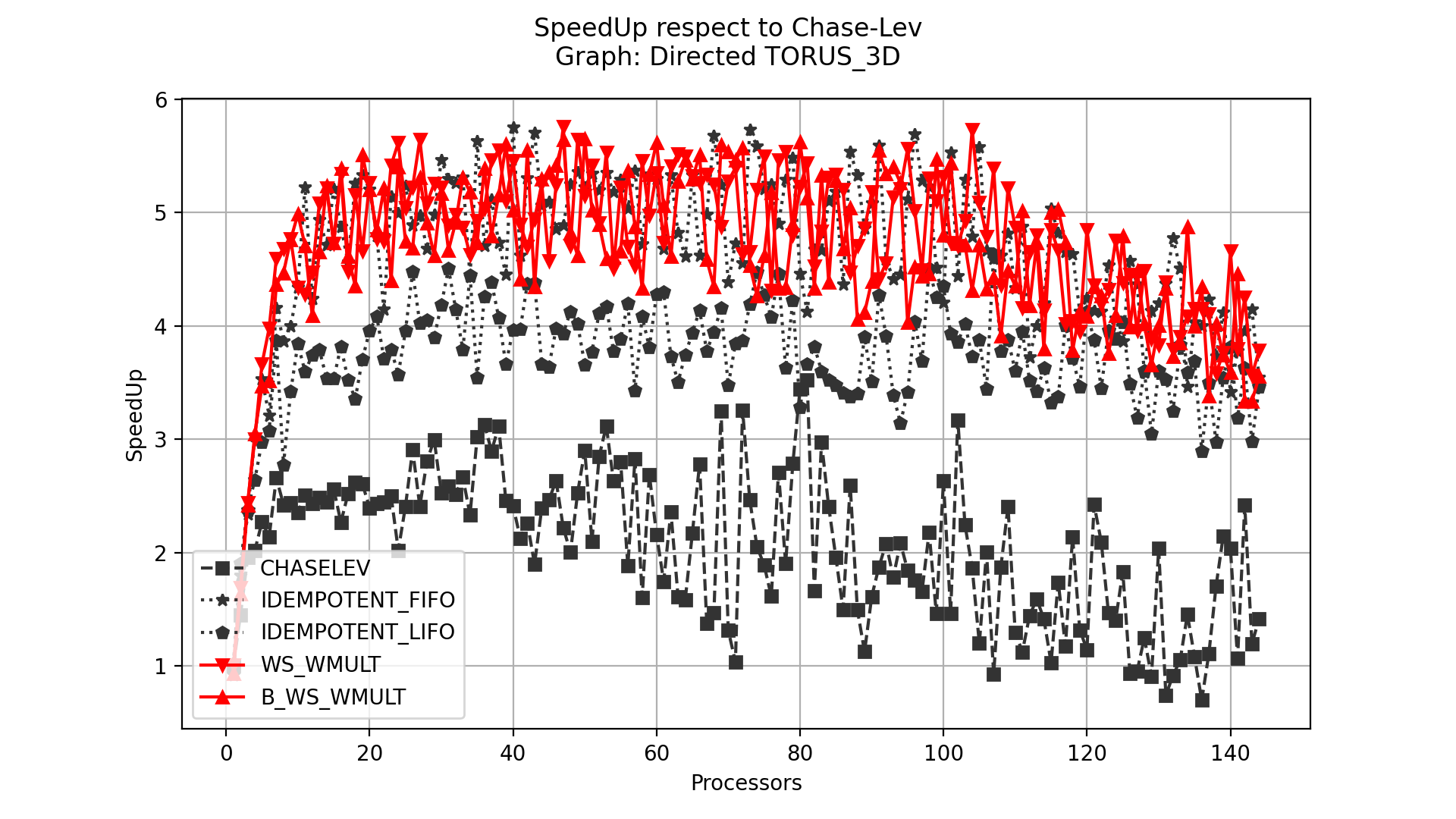}
    \label{cluster-directed-torus3d}
  }
  \subfloat[Subfigure 20][Speed up for undirected torus 3D graph on Intel
  Xeon.] {
    \includegraphics[scale=0.33]{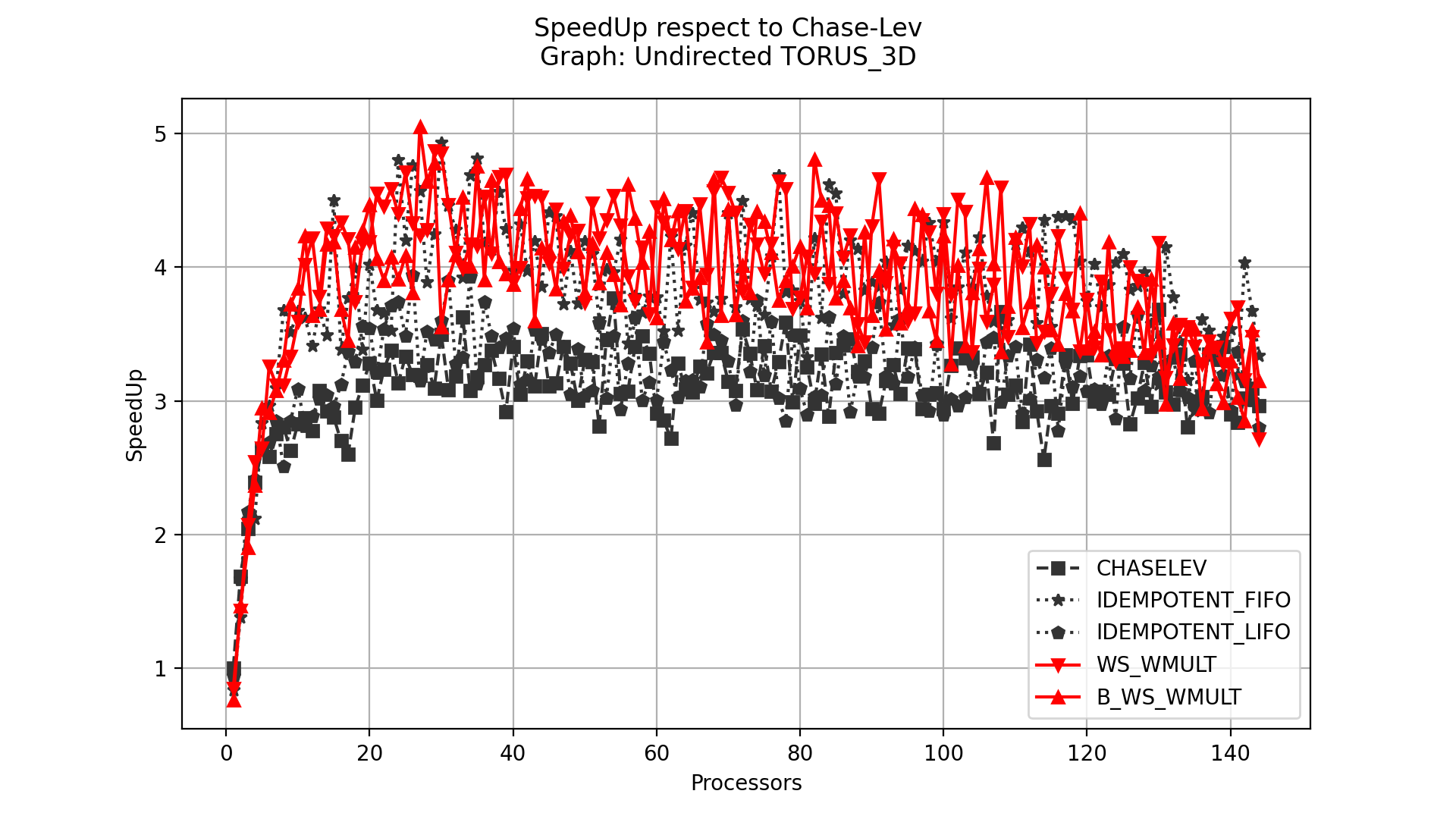}
    \label{cluster-undirected-torus3d}
  }
  \qquad
  \subfloat[Subfigure 21][Speed up for directed torus 3D40 graph on Intel
  Xeon.] {
    \includegraphics[scale=0.33]{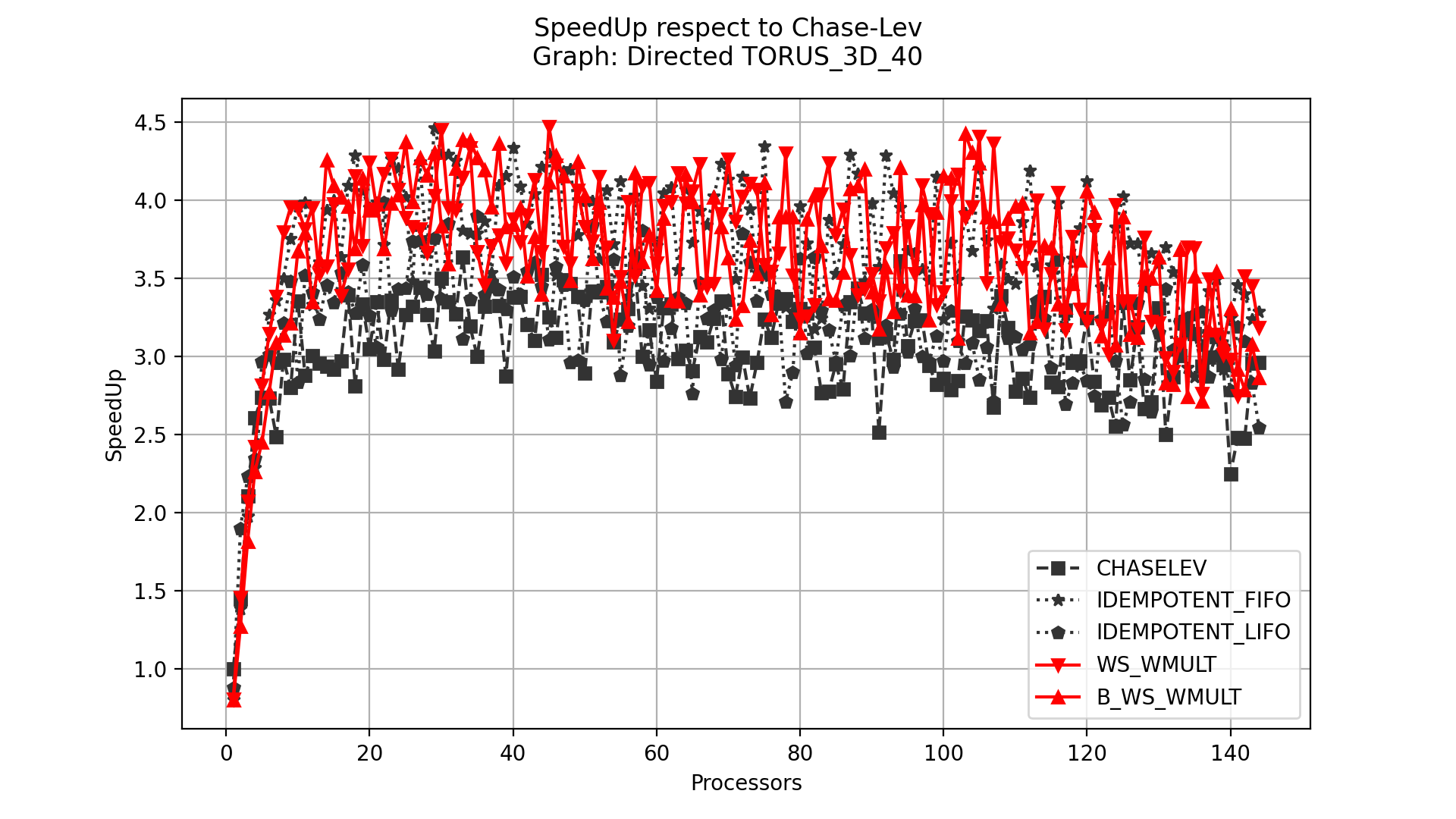}
    \label{cluster-directed-torus3d40}
  }
  \subfloat[Subfigure 22][Speed up for undirected torus 3D40 graph on Intel
  Xeon.] {
    \includegraphics[scale=0.33]{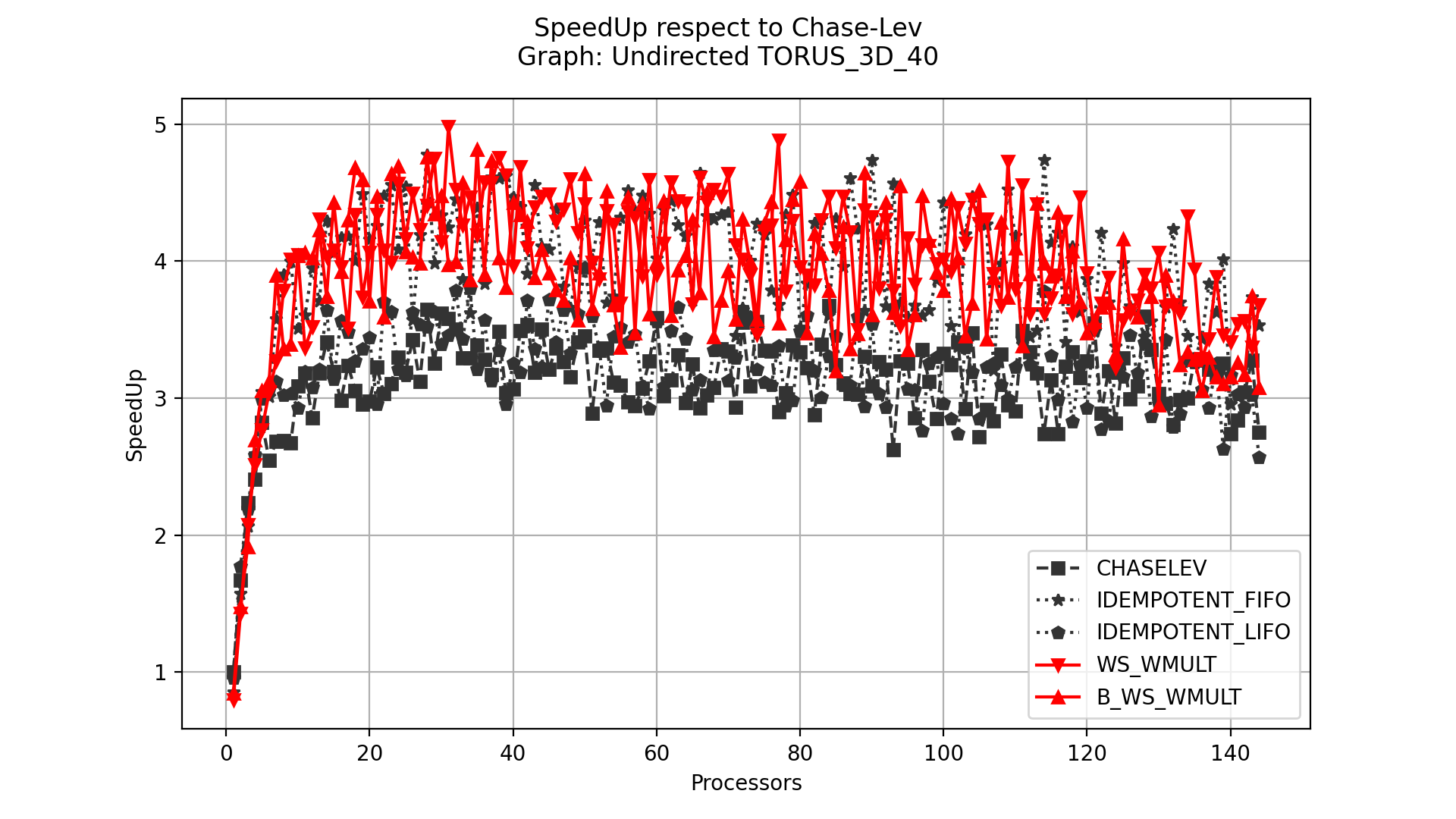}
    \label{cluster-undirected-torus3d40}
  }
\end{center}
    \caption{Speedups for 3D Torus on Intel Xeon}
\end{figure}

\section{Final Discussion}
\label{sec-final-discussion}

We have studied two relaxations for work-stealing, called multiplicity
and weak multiplicity. Both of them allow a task to be extracted by
more than one \Take/\Steal operation but each process can take the
same task at most once; however, the relaxation can arise only in
presence of concurrency. We presented fully \R/\W wait-free algorithms
for the relaxations. All algorithms are devoid of \RAW synchronization
patterns and the algorithm for the weak multiplicity case is also
fully fence-free with constant step complexity.  To the best of our
knowledge, this is the first time work-stealing algorithms with these
properties have been proposed, evading the impossibility result
in~\cite{AGHKMV11} in all operations.  From the theoretical
perspective of the consensus number hierarchy, we have thus shown that
work-stealing with multiplicity and weak multiplicity lays at the
lowest level of the hierarchy~\cite{H91}.

We have performed some an experimental evaluation comparing our
work-stealing solutions to THE Cilk, Chase-Lev and idempotent
work-stealing algorithms.  The experiments show that our \NCWSM
algorithm exhibits a better performance than the previously mentioned
algorithms, while its bounded version, \BNCWSM, with a \SWAP-based
\Steal operation, shows a
lower performance than \NCWSM, but keeps a competitive
performance with respect to the other algorithms.

All our results together show that one of the simplest synchronization
mechanisms suffices to solve non-trivial coordination problems,
particularly, computing a spanning tree.

For future research, we are interested in developing algorithms for
work-stealing with multiplicity and weak multiplicity that
insert/extract tasks in orders different from FIFO.
Also, it is interesting to explore if the techniques in our algorithms
can be applied to efficiently solve relaxed versions of other concurrent objects.
For example, it is worth to explore if a multi-enqueuer multi-dequeuer queue
with multiplicity (resp. weak multiplicity) can be obtained by manipulating
the tail with a \MaxReg (resp. \RangeMaxReg) object, like the head
is manipulated in our algorithms.

\bibliographystyle{plainurl}
\bibliography{references}

\begin{thebibliography}{10}

\bibitem{AF20}
Dolev Adas and Roy Friedman.
\newblock Brief announcement: Jiffy: {A} fast, memory efficient, wait-free
  multi-producers single-consumer queue.
\newblock In {\em 34th International Symposium on Distributed Computing, {DISC}
  2020, October 12-16, 2020, Virtual Conference}, pages 50:1--50:3, 2020.
\newblock \href {https://doi.org/10.4230/LIPIcs.DISC.2020.50}
  {\path{doi:10.4230/LIPIcs.DISC.2020.50}}.

\bibitem{AKY10}
Yehuda Afek, Guy Korland, and Eitan Yanovsky.
\newblock Quasi-linearizability: Relaxed consistency for improved concurrency.
\newblock In {\em Principles of Distributed Systems - 14th International
  Conference, {OPODIS} 2010, Tozeur, Tunisia, December 14-17, 2010.
  Proceedings}, pages 395--410, 2010.
\newblock \href {https://doi.org/10.1007/978-3-642-17653-1\_29}
  {\path{doi:10.1007/978-3-642-17653-1\_29}}.

\bibitem{AAC12}
James Aspnes, Hagit Attiya, and Keren Censor{-}Hillel.
\newblock Polylogarithmic concurrent data structures from monotone circuits.
\newblock {\em J. {ACM}}, 59(1):2:1--2:24, 2012.
\newblock \href {https://doi.org/10.1145/2108242.2108244}
  {\path{doi:10.1145/2108242.2108244}}.

\bibitem{AGHKMV11}
Hagit Attiya, Rachid Guerraoui, Danny Hendler, Petr Kuznetsov, Maged~M.
  Michael, and Martin~T. Vechev.
\newblock Laws of order: expensive synchronization in concurrent algorithms
  cannot be eliminated.
\newblock In {\em Proceedings of the 38th {ACM} {SIGPLAN-SIGACT} Symposium on
  Principles of Programming Languages, {POPL} 2011, Austin, TX, USA, January
  26-28, 2011}, pages 487--498, 2011.
\newblock \href {https://doi.org/10.1145/1926385.1926442}
  {\path{doi:10.1145/1926385.1926442}}.

\bibitem{ACDHLMTUZ09}
Eduard Ayguad{\'{e}}, Nawal Copty, Alejandro Duran, Jay Hoeflinger, Yuan Lin,
  Federico Massaioli, Xavier Teruel, Priya Unnikrishnan, and Guansong Zhang.
\newblock The design of openmp tasks.
\newblock {\em {IEEE} Trans. Parallel Distrib. Syst.}, 20(3):404--418, 2009.
\newblock \href {https://doi.org/10.1109/TPDS.2008.105}
  {\path{doi:10.1109/TPDS.2008.105}}.

\bibitem{BaderC04}
David~A. Bader and Guojing Cong.
\newblock A fast, parallel spanning tree algorithm for symmetric
  multiprocessors.
\newblock In {\em 18th International Parallel and Distributed Processing
  Symposium {(IPDPS} 2004), {CD-ROM} / Abstracts Proceedings, 26-30 April 2004,
  Santa Fe, New Mexico, {USA}}. {IEEE} Computer Society, 2004.
\newblock \href {https://doi.org/10.1109/IPDPS.2004.1302951}
  {\path{doi:10.1109/IPDPS.2004.1302951}}.

\bibitem{BJKLRZ95}
Robert~D. Blumofe, Christopher~F. Joerg, Bradley~C. Kuszmaul, Charles~E.
  Leiserson, Keith~H. Randall, and Yuli Zhou.
\newblock Cilk: An efficient multithreaded runtime system.
\newblock In {\em Proceedings of the Fifth {ACM} {SIGPLAN} Symposium on
  Principles {\&} Practice of Parallel Programming (PPOPP), Santa Barbara,
  California, USA, July 19-21, 1995}, pages 207--216, 1995.
\newblock \href {https://doi.org/10.1145/209936.209958}
  {\path{doi:10.1145/209936.209958}}.

\bibitem{CRR18}
Armando Casta{\~{n}}eda, Sergio Rajsbaum, and Michel Raynal.
\newblock Unifying concurrent objects and distributed tasks:
  Interval-linearizability.
\newblock {\em J. {ACM}}, 65(6):45:1--45:42, 2018.
\newblock \href {https://doi.org/10.1145/3266457} {\path{doi:10.1145/3266457}}.

\bibitem{CRR20}
Armando Casta{\~{n}}eda, Sergio Rajsbaum, and Michel Raynal.
\newblock Relaxed queues and stacks from read/write operations.
\newblock In {\em 24th International Conference on Principles of Distributed
  Systems, {OPODIS} 2020, December 14-16, 2020, Strasbourg, France (Virtual
  Conference)}, pages 13:1--13:19, 2020.
\newblock \href {https://doi.org/10.4230/LIPIcs.OPODIS.2020.13}
  {\path{doi:10.4230/LIPIcs.OPODIS.2020.13}}.

\bibitem{CGSDKEPS05}
Philippe Charles, Christian Grothoff, Vijay~A. Saraswat, Christopher Donawa,
  Allan Kielstra, Kemal Ebcioglu, Christoph von Praun, and Vivek Sarkar.
\newblock {X10:} an object-oriented approach to non-uniform cluster computing.
\newblock In {\em Proceedings of the 20th Annual {ACM} {SIGPLAN} Conference on
  Object-Oriented Programming, Systems, Languages, and Applications, {OOPSLA}
  2005, October 16-20, 2005, San Diego, CA, {USA}}, pages 519--538, 2005.
\newblock \href {https://doi.org/10.1145/1094811.1094852}
  {\path{doi:10.1145/1094811.1094852}}.

\bibitem{CL05}
David Chase and Yossi Lev.
\newblock Dynamic circular work-stealing deque.
\newblock In {\em {SPAA} 2005: Proceedings of the 17th Annual {ACM} Symposium
  on Parallelism in Algorithms and Architectures, July 18-20, 2005, Las Vegas,
  Nevada, {USA}}, pages 21--28, 2005.
\newblock \href {https://doi.org/10.1145/1073970.1073974}
  {\path{doi:10.1145/1073970.1073974}}.

\bibitem{D04}
Matei David.
\newblock A single-enqueuer wait-free queue implementation.
\newblock In {\em Distributed Computing, 18th International Conference, {DISC}
  2004, Amsterdam, The Netherlands, October 4-7, 2004, Proceedings}, pages
  132--143, 2004.
\newblock \href {https://doi.org/10.1007/978-3-540-30186-8\_10}
  {\path{doi:10.1007/978-3-540-30186-8\_10}}.

\bibitem{FDSZ01}
Christine~H. Flood, David Detlefs, Nir Shavit, and Xiolan Zhang.
\newblock Parallel garbage collection for shared memory multiprocessors.
\newblock In {\em Proceedings of the 1st Java Virtual Machine Research and
  Technology Symposium, April 23-24, 2001, Monterey, CA, {USA}}, 2001.
\newblock URL:
  \url{http://www.usenix.org/publications/library/proceedings/jvm01/full\_papers/flood/flood.pdf}.

\bibitem{FLR98}
Matteo Frigo, Charles~E. Leiserson, and Keith~H. Randall.
\newblock The implementation of the cilk-5 multithreaded language.
\newblock In {\em Proceedings of the {ACM} {SIGPLAN} '98 Conference on
  Programming Language Design and Implementation (PLDI), Montreal, Canada, June
  17-19, 1998}, pages 212--223, 1998.
\newblock \href {https://doi.org/10.1145/277650.277725}
  {\path{doi:10.1145/277650.277725}}.

\bibitem{HLMS06}
Danny Hendler, Yossi Lev, Mark Moir, and Nir Shavit.
\newblock A dynamic-sized nonblocking work stealing deque.
\newblock {\em Distributed Comput.}, 18(3):189--207, 2006.
\newblock \href {https://doi.org/10.1007/s00446-005-0144-5}
  {\path{doi:10.1007/s00446-005-0144-5}}.

\bibitem{HS02}
Danny Hendler and Nir Shavit.
\newblock Non-blocking steal-half work queues.
\newblock In {\em Proceedings of the Twenty-First Annual {ACM} Symposium on
  Principles of Distributed Computing, {PODC} 2002, Monterey, California, USA,
  July 21-24, 2002}, pages 280--289, 2002.
\newblock \href {https://doi.org/10.1145/571825.571876}
  {\path{doi:10.1145/571825.571876}}.

\bibitem{H91a}
Maurice Herlihy.
\newblock Impossibility results for asynchronous {PRAM} (extended abstract).
\newblock In {\em Proceedings of the 3rd Annual {ACM} Symposium on Parallel
  Algorithms and Architectures, {SPAA} '91, Hilton Head, South Carolina, USA,
  July 21-24, 1991}, pages 327--336, 1991.
\newblock \href {https://doi.org/10.1145/113379.113409}
  {\path{doi:10.1145/113379.113409}}.

\bibitem{H91}
Maurice Herlihy.
\newblock Wait-free synchronization.
\newblock {\em {ACM} Trans. Program. Lang. Syst.}, 13(1):124--149, 1991.
\newblock \href {https://doi.org/10.1145/114005.102808}
  {\path{doi:10.1145/114005.102808}}.

\bibitem{HS08}
Maurice Herlihy and Nir Shavit.
\newblock {\em The art of multiprocessor programming}.
\newblock Morgan Kaufmann, 2008.

\bibitem{HW90}
Maurice~P. Herlihy and Jeannette~M. Wing.
\newblock Linearizability: A correctness condition for concurrent objects.
\newblock {\em ACM Transactions on Programming Languages and Systems},
  12(3):463--492, July 1990.

\bibitem{JP05}
Prasad Jayanti and Srdjan Petrovic.
\newblock Logarithmic-time single deleter, multiple inserter wait-free queues
  and stacks.
\newblock In {\em {FSTTCS} 2005: Foundations of Software Technology and
  Theoretical Computer Science, 25th International Conference, Hyderabad,
  India, December 15-18, 2005, Proceedings}, pages 408--419, 2005.
\newblock \href {https://doi.org/10.1007/11590156\_33}
  {\path{doi:10.1007/11590156\_33}}.

\bibitem{JTT00}
Prasad Jayanti, King Tan, and Sam Toueg.
\newblock Time and space lower bounds for nonblocking implementations.
\newblock {\em {SIAM} J. Comput.}, 30(2):438--456, 2000.
\newblock \href {https://doi.org/10.1137/S0097539797317299}
  {\path{doi:10.1137/S0097539797317299}}.

\bibitem{L00}
Doug Lea.
\newblock A java fork/join framework.
\newblock In {\em Proceedings of the {ACM} 2000 Java Grande Conference, San
  Francisco, CA, USA, June 3-5, 2000}, pages 36--43, 2000.
\newblock \href {https://doi.org/10.1145/337449.337465}
  {\path{doi:10.1145/337449.337465}}.

\bibitem{MlVS09}
Maged~M. Michael, Martin~T. Vechev, and Vijay~A. Saraswat.
\newblock Idempotent work stealing.
\newblock In {\em Proceedings of the 14th {ACM} {SIGPLAN} Symposium on
  Principles and Practice of Parallel Programming, {PPOPP} 2009, Raleigh, NC,
  USA, February 14-18, 2009}, pages 45--54, 2009.
\newblock \href {https://doi.org/10.1145/1504176.1504186}
  {\path{doi:10.1145/1504176.1504186}}.

\bibitem{MA14}
Adam Morrison and Yehuda Afek.
\newblock Fence-free work stealing on bounded {TSO} processors.
\newblock In {\em Architectural Support for Programming Languages and Operating
  Systems, {ASPLOS} '14, Salt Lake City, UT, USA, March 1-5, 2014}, pages
  413--426, 2014.
\newblock \href {https://doi.org/10.1145/2541940.2541987}
  {\path{doi:10.1145/2541940.2541987}}.

\bibitem{N94}
Gil Neiger.
\newblock Set-linearizability.
\newblock In {\em Proceedings of the Thirteenth Annual {ACM} Symposium on
  Principles of Distributed Computing, Los Angeles, California, USA, August
  14-17, 1994}, page 396, 1994.
\newblock \href {https://doi.org/10.1145/197917.198176}
  {\path{doi:10.1145/197917.198176}}.

\bibitem{RRPBK07}
Colby Ranger, Ramanan Raghuraman, Arun Penmetsa, Gary~R. Bradski, and Christos
  Kozyrakis.
\newblock Evaluating mapreduce for multi-core and multiprocessor systems.
\newblock In {\em 13st International Conference on High-Performance Computer
  Architecture {(HPCA-13} 2007), 10-14 February 2007, Phoenix, Arizona, {USA}},
  pages 13--24, 2007.
\newblock \href {https://doi.org/10.1109/HPCA.2007.346181}
  {\path{doi:10.1109/HPCA.2007.346181}}.

\bibitem{SSONM10}
Peter Sewell, Susmit Sarkar, Scott Owens, Francesco~Zappa Nardelli, and
  Magnus~O. Myreen.
\newblock x86-tso: a rigorous and usable programmer's model for x86
  multiprocessors.
\newblock {\em Commun. {ACM}}, 53(7):89--97, 2010.
\newblock \href {https://doi.org/10.1145/1785414.1785443}
  {\path{doi:10.1145/1785414.1785443}}.

\bibitem{YM16}
Chaoran Yang and John~M. Mellor{-}Crummey.
\newblock A wait-free queue as fast as fetch-and-add.
\newblock In {\em Proceedings of the 21st {ACM} {SIGPLAN} Symposium on
  Principles and Practice of Parallel Programming, PPoPP 2016, Barcelona,
  Spain, March 12-16, 2016}, pages 16:1--16:13, 2016.
\newblock \href {https://doi.org/10.1145/2851141.2851168}
  {\path{doi:10.1145/2851141.2851168}}.

\end{thebibliography}

\end{document}